\documentclass[a4paper,amsart]{article}
\usepackage{textcomp}
\usepackage[utf8]{inputenc}
\usepackage[T1]{fontenc}
\usepackage[margin=35mm]{geometry}
\usepackage{amsmath, amsthm, amssymb}
\usepackage{graphicx}
\usepackage{enumerate}
\usepackage[textsize=small, textwidth=2cm, color=yellow]{todonotes}
\usepackage{thmtools}
\usepackage{thm-restate}
\declaretheorem[name=Theorem, numberwithin=section]{theorem}
\declaretheorem[name=Lemma, sibling=theorem]{lemma}

\declaretheorem[name=Claim, sibling=theorem]{claim}

\declaretheorem[name=Observation, style=remark, sibling=theorem]{observation}

\usepackage[colorlinks=true, citecolor=red]{hyperref}
\usepackage{float}
\usepackage[ruled,vlined]{algorithm2e}
\def\cqedsymbol{\ifmmode$\lrcorner$\else{\unskip\nobreak\hfil
\penalty50\hskip1em\null\nobreak\hfil$\lrcorner$
\parfillskip=0pt\finalhyphendemerits=0\endgraf}\fi} 
\newcommand{\cqed}{\renewcommand{\qed}{\cqedsymbol}}
\def\dd{\hbox{-}}
\usepackage{graphicx,subcaption}
\usepackage{xcolor}
\usetikzlibrary{snakes}
\usetikzlibrary{calc}
\definecolor{lila}{HTML}{7F00FF}
\definecolor{grun}{HTML}{00994C}
\definecolor{darkorange}{HTML}{D55E00}
\definecolor{rosa}{HTML}{CC79A7}
\definecolor{gelb}{HTML}{F0E442}
\tikzstyle{snode}=[shape= circle, draw=red, ultra thick, fill=red!30]
\tikzstyle{special node} = [shape= circle, draw=blue, ultra thick, fill=blue!10]
\tikzstyle{blue node} = [special node]
\tikzstyle{yellow node} = [shape= circle, draw=gelb, ultra thick, fill=gelb!30]
\tikzstyle{red node} = [snode]
\tikzstyle{green node} =  [shape= circle, draw=grun, ultra thick, fill=grun!20]
\tikzstyle{maybe node} = [shape = circle, draw=black, dotted, thick]
\tikzstyle{maybe new node} =[shape = circle, draw=blue, dotted, ultra thick, fill=blue!10]
\tikzstyle{maybe equal} = [dotted] %make thick?
\tikzstyle{normal node} =[shape=circle,draw=black, fill = white]
\tikzstyle{new edge} = [blue, ultra thick]
\tikzstyle{maybe new edge}=[new edge, decorate]
\tikzstyle{non edge} = [red, dashed]
\tikzset{decoration={snake,amplitude=.7mm,segment length=2.5mm,
                       post length=0mm,pre length=0mm}}
\tikzstyle{maybe edge} = [decorate ]
\tikzstyle{dot} = [shape = circle, draw = black, fill = black, scale=1.1]

\usepackage[toc,page]{appendix}
\title{A tight local algorithm for the minimum dominating set problem in outerplanar graphs}
\usepackage{authblk}

\author[1]{Marthe Bonamy\thanks{\texttt{marthe.bonamy@u-bordeaux.fr}, supported by the ANR project GrR (\textsc{ANR-18-CE40-0032}).}}
\author[2]{Linda Cook\thanks{\texttt{lindacook@ibs.re.kr}, supported by AFOSR grant A9550-19-1-0187, NSF grant DMS-1800053 and the Institute for Basic Science (IBS-R029-C1).}}
\author[3]{Carla Groenland\thanks{\texttt{c.e.groenland@uu.nl}, supported by the European Research Council Horizon 2020 project CRACKNP (grant agreement no. 853234).}}
\author[4]{Alexandra Wesolek\thanks{\texttt{agwesole@sfu.ca}, supported by the Vanier Canada Graduate Scholarships program.}}
\affil[1]{CNRS, LaBRI, Universit\'{e} de Bordeaux, Bordeaux, France.}
\affil[2]{Discrete Mathematics Group, Institute for Basic Science (IBS), Daejeon, Republic of Korea.}
\affil[3]{Utrecht University, Utrecht, Netherlands.}
\affil[4]{Department of Mathematics, Simon Fraser University, Burnaby, BC, Canada.}
\date{\today}
\begin{document}
\maketitle

\begin{abstract}
We show that there is a deterministic local algorithm (constant-time distributed graph algorithm) that finds a 5-approximation of a minimum dominating set on outerplanar graphs. We show there is no such algorithm that finds a $(5-\varepsilon)$-approximation, for any $\varepsilon>0$. Our algorithm only requires knowledge of the degree of a vertex and of its neighbors, so that large messages and unique identifiers are not needed.
\end{abstract}

\section{Introduction}\label{sec:intro}
Given a sparse graph class, how well can we approximate the size of the minimum dominating set (MDS) in the graph using a constant number of rounds in the LOCAL model? A \emph{dominating set} of a graph $G=(V,E)$ is a set $S \subseteq V$ such that every vertex in $V \setminus S$ has a neighbor in $S$. Given a graph $G$ and an integer $k$, deciding whether $G$ has a dominating set of size at most $k$ is NP-complete even when restricting to planar graphs of maximum degree three \cite{garey1979computers}. Moreover, the size of the MDS is NP-hard to approximate within a constant factor (for general graphs) \cite{RanRaz}. The practical applications of MDS are diverse but almost always involve large networks \cite{alipour2020distributed}, and it is therefore natural to turn to the the distributed setting. No constant factor approximation of the MDS is possible using a sub-linear number of rounds in the LOCAL model~\cite{kuhn2016local}, and so various structural restrictions have been considered on the graph classes with the hope of finding more positive results  (see~\cite{feuilloley2020bibliography} for an overview). 

Planar graphs are a hallmark case. For planar graphs, guaranteeing that some constant factor approximation can be achieved is already highly non-trivial~\cite{czygrinow2008fast,lenzen2008can}. The current best known upper-bound is $52$~\cite{wawrzyniak2014strengthened}, while the best lower-bound is $7$~\cite{hilke2014brief}. Substantial work has focused on generalizing the fact that some constant factor approximation is possible to more general classes of sparse graphs, like graphs that can be embedded on a given surface, or more recently graphs of bounded expansion~\cite{akhoondian2018distributed,akhoondian2016local, amiri2019distributed, czygrinow18,czygrinow19,kublenz2020distributed}. Tight bounds currently seem out of reach in those more general contexts.

In this paper we focus instead on restricted subclasses of planar graphs. Better approximation ratios can be obtained with additional structural assumptions: $32$ if the planar graph contains no triangle~\cite{alipour2020distributed} and $18$ if the planar graph contains no cycle of length four~\cite{alipour2020local}. These bounds are not tight, and in fact we expect they can be improved significantly. We are able to provide tight bounds for a different type of restriction: we consider planar graphs with no $K_{2,3}$-minor or $K_4$-minor\footnote{For any integer $n \geq 1$, $K_n$ denotes the complete graph on $n$ vertices. For integers $n, m \geq 1$, $K_{n, m}$ denotes the complete bipartite graph with partite classes of size $n$ and $m$.}, i.e. \emph{outerplanar} graphs. Outerplanar graphs can alternatively be defined as planar graphs that can be embedded so that there is a special face which contains all vertices in its boundary.

Outerplanar graphs are a natural intermediary graph class between planar graphs and forests. A planar graph on $n$ vertices contains at most $3n-6$ edges, and a forest on $n$ vertices contains at most $n-1$ edges; an outerplanar graph on $n$ vertices contains at most $2n-3$ edges. Every planar graph can be decomposed into three forests~\cite{nash1964decomposition}; it can also be decomposed into two outerplanar graphs~\cite{gonccalves2005edge}.

For planar graphs, as discussed above, we are far from a good understanding of how to optimally approximate Minimum Dominating Set in $O(1)$ rounds. Let us discuss the case of forests, as it is of very relevant to the outerplanar graph case. For forests, a trivial algorithm yields a $3$-approximation: it suffices to take all vertices of degree at least $2$ in the solution, as well as vertices with no neighbor of degree at least $2$ (that is, isolated vertices and isolated edges). The output is clearly a dominating set, and the proof that it is at most three times as big as the optimal solution is rather straightforward. In fact, the trivial algorithm is tight because of the case of long paths. Indeed, no constant-time algorithm can avoid taking all but a sub-linear number of vertices of a long path, while there is a dominating set containing only a third of the vertices.

\subsection*{Our contribution}
We prove that a similarly trivial algorithm (as the one described for forests above) works to obtain a 5-approximation of MDS for outerplanar graphs in the LOCAL model.
\bigskip

\begin{algorithm}[H]\label{alg:1}
\SetAlgoLined
\KwIn{An outerplanar graph $G$}
\KwResult{A set $S \subseteq V(G)$ that dominates $G$}
 In the first round, every vertex computes its degree and sends it to its neighbors\;
 $S:=\{\textrm{Vertices of degree $\geq 4$}\}\cup\{\textrm{Vertices with no neighbor of degree $\geq 4$}\}$\;
 \caption{A local algorithm to compute a dominating set in outerplanar graphs}
\end{algorithm}
\bigskip
It is easy to check that the algorithm indeed outputs a dominating set. It is significantly harder to argue that the resulting dominating set is at most $5$ times as big as one of minimum size. To do that, we delve into a rather intricate analysis of the behavior of a hypothetical counterexample, borrowing tricks from structural graph theory (see Lemma~\ref{lem:outernew}).

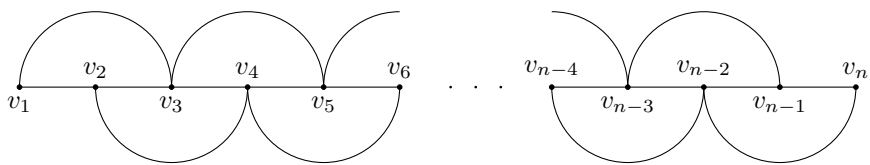
\begin{figure}[!h]
    \centering
    \begin{tikzpicture}
    
    \foreach \i in {1,3,5} {
    \draw [fill] (\i-1,0) circle [radius=0.03] node[below] {$v_{\i}$};
    }
    \foreach \i in {2,4,6} {
    \draw [fill] (\i-1,0) circle [radius=0.03] node[above] {$v_{\i}$};
    }
    
       \foreach \i in {1,3} {
    \draw [fill] (11-\i,0) circle [radius=0.03] node[below] {$v_{n-\i}$};
    }
    
    \foreach \i in {2,4} {
    \draw [fill] (11-\i,0) circle [radius=0.03] node[above] {$v_{n-\i}$};
    }
        \draw [fill] (11,0) circle [radius=0.03] node[above] {$v_{n}$};
    
    \foreach \i in {1,...,5} {
    \draw (\i-1,0) -- (\i,0);
    }
        \foreach \i in {2,4,10} {
    \draw (\i,0) arc (0:180:1);
    }
    \draw (4,0) arc (180:90:1);

     \draw (8,0) arc (0:90:1);
            \foreach \i in {1,3,7,9} {
    \draw (\i,0) arc (180:360:1);
    }
    
     \draw [fill] (5.67,0) circle [radius=0.01];
    \draw [fill] (6,0) circle [radius=0.01];
    \draw [fill] (6.33,0) circle [radius=0.01];
     \foreach \i in {8,...,11} {
    \draw (\i-1,0) -- (\i,0);
    }
    \end{tikzpicture}
    \caption{The graph $G^-_n$ is a path $v_1,\dots,v_n$ to which we add all edges between vertices of distance two. In this example $n$ is even.}
    \label{fig:LowerBoundGraph}
\end{figure}

 The proof that the bound of $5$ is tight for outerplanar graphs is similar to the proof that the bound of $3$ is tight for trees. Every graph in the family depicted in Figure~\ref{fig:LowerBoundGraph} is outerplanar, and every local algorithm that runs in a constant number of rounds selects all but a sub-linear number of vertices~\cite[pp.~87--88]{czygrinow2008fast}. Informally, all but a sub-linear number of vertices ``look the same'' -- see Section~\ref{sec:lowerbound} for more details, and~\cite{suomela2013survey} for an excellent survey of lower bounds. 

Our main result is the following.
\begin{theorem}\label{th:outerplanar}
There is an algorithm that computes a $5$-approximation of Minimum Dominating Set for outerplanar graphs in $O(1)$ rounds in the LOCAL model. This is tight, in the sense that no algorithm can compute a  $(5-\varepsilon)$-approximation with the same constraints, for any $\varepsilon>0$.
\end{theorem}

In other words, there is a trivial local algorithm for Minimum Dominating Set in outerplanar graphs that turns out to be tight. All the difficulty lies in arguing that the approximation factor is indeed correct.

We note that the algorithm is so trivial that every vertex only needs to send one bit of information to each of its neighbors (``I have degree at least $4$'' or ``I have degree at most $3$''). The network might be anonymous -- names are not useful beyond being able to count the number of neighbors, and the solution is extremely easy to update when there is a change in network. For contrast, in anonymous planar graphs the best known approximation ratio is 636~\cite{wawrzyniak2013brief}.

It is important to note that there is no hope for such a trivial algorithm in the case of planar graphs. Indeed, in Figure~\ref{fig:planarno}, we can see that for any $p$, no algorithm taking all vertices of degree $\geq p$ in the solution can yield a constant-factor approximation in planar graphs. However, the case of outerplanar graphs shows that the road to a better bound for planar graphs might go through finer structural analysis rather than smarter algorithms.

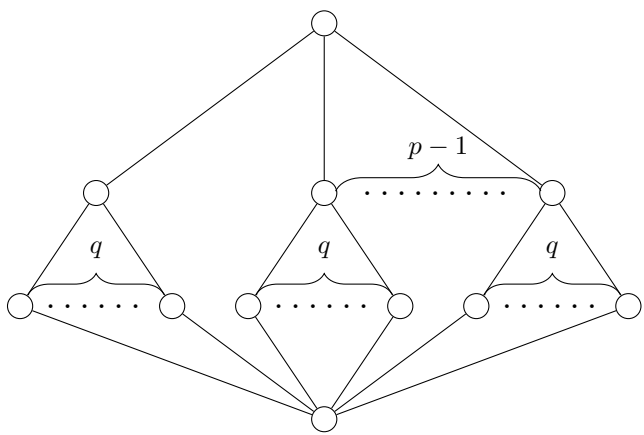
\begin{figure}[!h]
    \centering
    %\section{Linda 1}
\begin{tikzpicture}

\def\height{1.5};
\def\width{1};

\node[normal node] (a) at (4*\width, 2.5*\height){};

\node[normal node] (c_1) at (0,0) {};
\node[normal node] (c_2) at (2*\width,0){};
\node[normal node] (b_1) at (\width, \height){};
\draw (b_1) to (c_1);
\draw (b_1) to (c_2);
\draw (b_1) to (a);
\node at (\width, 0) {\Large $\dots \dots$};
\draw [decoration={brace,amplitude=8pt},decorate,] ($(c_1)+(.3em,1ex)$) -- ($(c_2)+(-.3em,1ex)$);
\node at (\width, 0.5*\height) {$q$};

\node[normal node] (c_3) at (3*\width,0) {};
\node[normal node] (c_4) at (5*\width,0){};
\node[normal node] (b_2) at (4*\width, \height){};
\draw (b_2) to (c_3);
\draw (b_2) to (c_4);
\draw (b_2) to (a);
\node at (4*\width, 0) {\Large $\dots \dots$};
\draw [decoration={brace,amplitude=8pt},decorate,] ($(c_3)+(.3em,1ex)$) -- ($(c_4)+(-.3em,1ex)$);
\node at (4*\width, 0.5*\height) {$q$};

\node[normal node] (c_5) at (6*\width,0) {};
\node[normal node] (c_6) at (8*\width,0){};
\node[normal node] (b_3) at (7*\width, \height){};
\draw (b_3) to (c_5);
\draw (b_3) to (c_6);
\draw (b_3) to (a);
\node at (7*\width, 0) {\Large $\dots \dots$};
\draw [decoration={brace,amplitude=8pt},decorate,] ($(c_5)+(.3em,1ex)$) -- ($(c_6)+(-.3em,1ex)$);
\node at (7*\width, 0.5*\height) {$q$};

\node at (5.5*\width, \height) {\Large $\dots \dots \dots$};
\draw [decoration={brace,amplitude=10pt},decorate,] ($(b_2)+(.43em,.18ex)$) -- ($(b_3)+(-.43em,.18ex)$);
\node at (5.5*\width,\height + .6) {$p-1$};

\node[normal node] (d) at (4*\width, -\height) {};
\draw (c_1) to (d);
\draw (c_2) to (d);
\draw (c_3) to (d);
\draw (c_4) to (d);
\draw (c_5) to (d);
\draw (c_6) to (d);
\end{tikzpicture}
    \caption{For any $p, q \in \mathbb{N}$, there is a planar graph $G_{p,q}$ which admits a dominating set of size $2$ such that $|\{\textrm{Vertices with degree }\geq q\}|\geq p$}.
    \label{fig:planarno}
\end{figure}

\subsection*{Definitions and notation}

For a vertex set $A \subseteq V$, let $G[A]$ denote the induced subgraph of $G$ with vertex set $A$. Let $E(A)$ denote set of edges of $G[A]$. For vertex sets $A, B \subseteq G$, let $E(A, B)$ denote the set of edges in $G$ with one end in $A$ and the other end in $B$. We write $G\setminus e$ for the graph in which the edge $e$ is removed from the edge set of $G$. For a set $P\subseteq V$ inducing a connected subgraph, we write $G/P$ for the graph obtained by contracting the set $P$: we replace the vertices in $P$ with a new vertex $v_P$, which is adjacent to $u\in V\setminus P$ if and only if $u$ has some neighbor in $P$. 
For a set $X$ of vertices, we let $N[X]$ denote the set $ X \cup \cup_{x \in X}N(x)$ and we let $N(X)$ denote the set $N[X] \setminus X$.
If $x_1, x_2, \dots, x_k$ are the elements of $X$, we may also denote $N[X]$ and $N(X)$ as $N[x_1, x_2, \dots, x_k]$ and $N(x_1, x_2, \dots, x_k)$, respectively.

Given a graph $G$, let $V_{4+}(G)$ denote the set of vertices of degree at least $4$ in $G$, and let $V^*(G)$ denote the set $V(G)\setminus N[V_{4+}(G)]$. In other words, $V^*(G)$ is the set of vertices of degree at most $3$ in $G$ which only have neighbors of degree at most $3$. 
For a graph $G$ and a dominating set $S$ of $G$, we denote $V_{4^+}(G)\setminus S$ by $B_S(G)$ and we denote $V^*(G) \setminus S$ by $D_S(G)$. We additionally let $A_S(G)$ denote the set $V(G) \setminus (S \cup D_S(G) \cup B_S(G))$. In situations where our choice of $G,S$ is not ambiguous we will simply write $B,D, A$ for $B_S(G)$, $D_S(G)$ and $A_S(G)$, respectively.
An overview of the notation is given in Table \ref{tab:notation}.\\
\begin{table}[htb]
\begin{tabular}{c|c|c|c}
    $v$ is an element of  & $\text{deg}(v)$ & degrees of neighbors of $v$ & further restrictions \\
    \hline
     $V_{4^+}(G)$ &  $\geq 4$ & arbitrary & -\\
     $B_S(G)$ &  $\geq 4$ & arbitrary & $v\notin S$\\
     $V^*(G)$ & $\leq 3$ & $\leq 3$ & -\\
     $D_S(G)$ &  $\leq 3$ & $\leq 3$ & $v\notin S$\\
     $A_S(G)$ & $\leq 3$ & at least one neighbor of degree $\geq 4$ & $v\notin S$\\
\end{tabular}
\caption{An overview of the notation used in Section \ref{sec:upperbound}.}
\label{tab:notation}
\end{table}

An outerplanar embedding of $G$ is an embedding in which a special \emph{outer face} contains all vertices in its boundary. 

We denote by $H_G(S)$ the multigraph with vertex set $S$, obtained from $G$ as follows. For every vertex $u$ in $V(G) \setminus S$, we select a neighbor $s(u) \in N(u) \cap S$, and contract the edge $\{u,s(u)\}$. Contrary to the contraction operation mentioned earlier, this may create parallel edges, but we delete all self-loops. The resulting multigraph inherits the set $S$ as its vertex set.  We refer to Figure $\ref{fig:HGraph}$ for an example. 
\begin{figure}[h]
\begin{tikzpicture}[scale=0.5]   
\node (before) at (0,0){
    \begin{tikzpicture}[scale=0.5]  
    \def\height{1.5};
    \def\width{2};
    
    \node[snode] (s_2) at (0, 0){$s_2$};
    \node[snode] (s_1) at (3*\width ,0){$s_1$};
    \node[snode] (s_3) at (5* \width, 0){$s_3$};
    \node[snode] (s_4) at (3*\width ,-3){$s_4$};
    
    \node[normal node] (b_1) at (3*\width, 2*\height) {$b_1$};
    \node[normal node] (a_2) at (2*\width, 0){$a_2$};
    
    \node[normal node] (a_1) at (1*\width, \height){$a_1$};
    \node[normal node] (d_1) at (1*\width, -\height){$d_1$};
    
    \node[normal node] (a_3) at (4*\width, \height){$a_3$};
    \node[normal node] (a_4) at (4*\width, -1*\height){$a_4$};
    
    \node[normal node] (d_2) at (3*\width ,-5){$d_2$};
    
    \draw (s_1) to (b_1);
    \draw (s_1) to (a_2);
    \draw (b_1) to (a_1);
    \draw (a_2) to (d_1);
    \draw (s_2) to (a_1);
    \draw (s_2) to (d_1);
    \draw (s_3) to (a_3);
    \draw (s_3) to (a_4);
    \draw (b_1) to (a_3);
    \draw (b_1) to (a_4);
    \draw (s_4) to (s_1);
    \draw (s_4) to (a_4);
    \draw (s_4) to (d_2);
    \end{tikzpicture}
    };
    \node(after) at (before.east)[anchor=east,xshift=7.5cm,yshift=-0.2cm]{
    \begin{tikzpicture}[scale=.5]
    \def\height{1.5};
    \def\width{2};
    \node[snode] (s_2) at (0, 0){$s_2$};
    \node[snode] (s_1) at (3*\width ,0){$s_1$};
    \node[snode] (s_3) at (5* \width, 0){$s_3$};
    \node[snode] (s_4) at (3*\width ,-3){$s_4$};
    \draw[bend left=20] (s_2) to (s_1);
    \draw[bend left=20] (s_1) to (s_2);
    \draw[bend left=20] (s_1) to (s_3);
    \draw[bend left=20] (s_3) to (s_1);
    \draw (s_1) to (s_4);
    \draw (s_3) to (s_4);
\end{tikzpicture}
};
    \end{tikzpicture}
    \caption{On the left is a graph $G$ with dominating set $S=\{s_1,s_2,s_3,s_4\}$. The vertex $s(u)$ is uniquely determined for all $u \neq a_4$. On the right is the graph $H_G(S)$ for $s(a_4)=s_3$. }
    \label{fig:HGraph}
\end{figure}
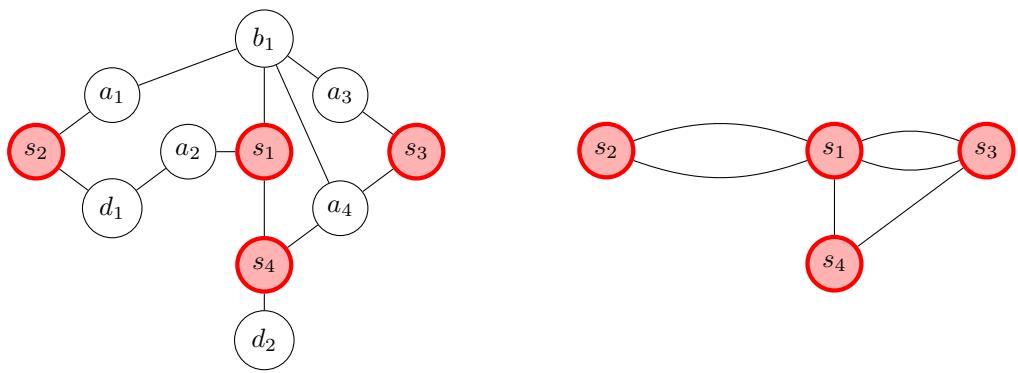

Note that $H_G(S)$ inherits an outerplanar embedding from $G$.
If the graph $G$ and the dominating set $S$ are clear, we will write $H$ for $H_G(S)$. Lemma $\ref{lem:outerweak}$ provides some intuition as to why the graph $H$ is useful.

\subsection*{Properties of outerplanar graphs}

Here we mention some standard but useful properties of outerplanar graphs.
A graph $H$ is a \emph{minor} of a graph $G$ if $H$ can be obtained from $G$ through a series of vertex or edge deletions and edge contractions. Alternatively, an $H$-minor of $G$ consists of a connected set $X_h\subseteq V(G)$ for each $h\in V(H)$ and a set of paths $\{P_{hh'}| \, hh'\in E(H)\}$, where $P_{hh'}$ is a path in $G$ between a vertex in $X_h$ and a vertex in $X_{h'}$, all of which are pairwise vertex-disjoint except for possibly their ends. Note that any minor of an outerplanar graph is outerplanar. 
%In other words, outerplanar graphs are closed under taking minors.
Neither $K_4$ nor $K_{2,3}$ can be drawn in the plane so that all vertices appear on the boundary of a special face. Therefore, outerplanar graphs are $K_4$-minor-free and $K_{2,3}$-minor-free. 

Any outerplanar graph $G$ satisfies $|E(G)|\leq 2|V(G)|-3$ by a simple application of Euler's formula. It follows immediately that every outerplanar graph contains a vertex of degree at most $3$, but a standard structural analysis guarantees that every outerplanar graph contains a vertex of degree at most $2$.

\section{Analysis of the approximation factor}
\label{sec:upperbound}
This section is devoted to proving the following result. (An overview of the relevant notation is given in Table \ref{tab:notation}.)

\begin{lemma}\label{lem:outernew}
For every outerplanar graph $G$, any dominating set $S$ of $G$ satisfies $|S|\geq \frac14(|B_S(G)|+|D_S(G)|)$.
\end{lemma}

We briefly argue that Lemma~\ref{lem:outernew} yields the desired result. Given an outerplanar graph, Algorithm~\ref{alg:1} outputs $V_{4^+}(G) \cup V^*(G)$ as a dominating set. To argue that it is a $5$-approximation of the Minimum Dominating Set problem, it suffices to prove that any dominating set $S$ of $G$ satisfies $|S|\geq \frac15 (|V_{4^+}(G) \cup V^*(G)|)$. For technical reasons, it is easier to bound $S$ as a function of the vertices in $V_{4^+}(G) \cup V^*(G)$ that are not in $S$, i.e.
$|S|\geq \frac14(|B_S(G)|+|D_S(G)|)$, which yields $|S|\geq \frac15 (|V_{4^+}(G) \cup V^*(G)|)$.

We prove the lemma by analyzing the structure of a ``smallest'' counterexample. A counterexample satisfies 
\[    |S| < \frac14 (|B_{S}(G)|+|D_S(G)|),
\]
 and we will choose one which minimizes $|S|$ and with respect to that maximizes $|B_S(G)|+| D_S(G)|$. For this, we need that $|B_S(G)|+|D_S(G)|$ is bounded in terms of $|S|$  by some constant, otherwise a counterexample maximizing $|B_S(G)|+|D_S(G)|$ might not exist since $|B_S(G)|+|D_S(G)|$ could be arbitrarily large. We therefore first prove the following much weaker result.
\begin{lemma}
\label{lem:outerweak}
For every outerplanar graph $G$, any dominating set $S$ of $G$ satisfies $|S|\geq \frac1{39}(|B_S(G)|+|D_S(G)|)$.
\end{lemma}
We did not try to optimize the constant $39$ and rather aim to get across some of the main ideas as clearly as possible. The proof shows the importance of the graph $H_G(S)$, which we will also use in the proof of Lemma $\ref{lem:outernew}$.
\begin{proof}[Proof of Lemma \ref{lem:outerweak}]
We may assume that the graph $G$ is connected; otherwise, we can repeat the same argument for each connected component of $G$. We fix an outerplanar embedding of $G$. For each $u \in V(G)\setminus S$ we select an arbitrary neighbor $s(u)\in N(u)\cap S$ that we contract it with (keeping parallel edges but removing self-loops), resulting in the multigraph $H_G(S)$ on the vertex set $S$. The key step in our proof is showing that $H_G(S)$ has bounded edge multiplicity. Indeed, every edge $s_1s_2$ in $H_G(S)$ is obtained from $G$ by contracting at least one vertex or from the edge $s_1s_2$ in $G$.
For $i\in \{1,2\}$, let $V_i$ be the set of vertices contracted to $s_i$ that gave an edge between $s_1$ and $s_2$ in $H_G(S)$. Since there is no $K_{2,3}$-minor in $G$ (as $G$ is outerplanar), we find $|V_1|\leq 2$ and $|V_2|\leq 2$. Any edge between $s_1$ and $s_2$ in $H_G(S)$ can now be associated with an edge between $\{s_1\}\cup V_1$ and $\{s_2\}\cup V_2$ in $G$, and hence edges in $H_G(S)$ have multiplicity at most $9$ (this is far from tight).

We derive that $|E(H_G(S))|\leq 9 |E(H')|$, where $H'$ is the simple graph underlying $H_G(S)$ (i.e. the simple graph obtained by letting  $s_1,s_2\in S$ be adjacent in $H'$ if and only if there is an edge between them in $H_G(S)$). Note that $H'$ is a minor of $G$. Since outerplanar graphs are closed under taking minors, the graph $H'$ is an outerplanar graph. It follows that $|E(H')|\leq 2|S|-3$. Combining both observations, we get $|E(H_G(S))|\leq 18 |S|$.

By outerplanarity, we have $|E(H_G(S))|\geq \frac12|B_S(G)|$. Indeed, each vertex $u\in B_S(G)$ has at most two common neighbors with $s(u)$ (otherwise there would be a $K_{2,3}$), hence $u$ has at least one neighbor $v$ such that $v \not\in N[s(u)]$. The edge $uv$ corresponds to an edge in  $E(H_G(S))$, hence each $u\in B_S(G)$ contributes at least half an edge to $E(H_G(S))$ (as $v$ could be also in $B_S(G)$).
We derive $\frac12|B_S(G)| \leq |E(H_G(S))|\leq 18 |S|$. We observe that $|D_S(G)| \leq 3|S|$: indeed, each vertex from $S$ is adjacent to at most $3$ vertices from $D_S(G)$, since any vertex adjacent to a vertex in $D_S(G)$ has degree at most $3$ by definition. We conclude that $|B_S(G)|+|D_S(G)|\leq (36+3)|S| = 39 |S|$. 
\end{proof}
In Lemma~\ref{lem:outerweak} we use that edges in $H$ have low multiplicity, from which we then obtain a bound on the size of $S$. In order to improve the bound from Lemma  \ref{lem:outerweak}, we dive into a deeper analysis of the graph $H$.
\begin{proof}[Proof of Lemma \ref{lem:outernew}]
We will consider a special counterexample $(G,S)$ (satisfying $|S|<\frac14(|B_S(G)|+|D_S(G)|)$) so that our counterexample has a structure we can deal with more easily than a general counterexample.
%In particular we will choose a counterexample $(G, S)$ amongst those that minimizes $|S|$ and with respect to that maximizes $|B_S(G) \cup D_S(G)|$ and with respect to that 
Namely, we assume that $(G,S)$ in order:
%minimizes $|S|$; maximizes $|B_S(G) \cup D_S(G)|$; minimizes $|E(B_S(G))|$; minimizes $|V(G)|$; maximizes $|E(S,N(S))|$; minimizes $|E(G)|$. 
minimizes $|S|$; maximizes $|B_S(G) \cup D_S(G)|$; minimizes $|E(G)|$.
%Note that this is well-defined since we established $|B_S(G)\cup D_S(G)| \leq 39 |S|$, and clearly $|E(S,N(S))|\leq |E(G)| \leq  2 |V(G)|$ by outerplanarity. 
Note that this is well-defined since we established $|B_S(G)\cup D_S(G)| \leq 39 |S|$,
Consequently, if a counterexample exists, then there exists one satisfying all of the above assumptions. More formally, we select a counterexample that is minimal for
\begin{align}
    (|S|,39|S|-|B_S(G) \cup D_S(G)|,|E(G)|)\tag{$\ddagger$}\label{eq:ddagger}
\end{align}
in the lexicographic order. Since all the elements in the triple are non-negative integers and their minimum is bounded below by zero, this is well-defined. (We remark that the parts indicated in gray were added to ensure the entries are non-negative; minimizing $39|S|-|B_S(G) \cup D_S(G)|$ comes down to maximizing $|B_S(G) \cup D_S(G)|$.)

While this approach is not entirely intuitive, the assumptions will prove to be extremely useful for simplifying the structure of $G$. For example, we can show that in a smallest counterexample that minimizes $\eqref{eq:ddagger}$,  $S$ is a stable set (Claim \ref{cl:Sstable}) and no two vertices in $S$ have a common neighbor (Claim \ref{cl:Snocommon}). In general, the Claims \ref{cl:DS} to \ref{cl:v1neqw1} show that such a minimal counterexample $G$ has strong structural properties, by showing that otherwise we could delete some vertices and edges, or contract edges, and find a smaller counterexample. 

Informally, for any vertex from $S$ that we remove from the graph, we may decrease $|B_S(G)\cup D_S(G)|$ by 4 while maintaining $|S|<\frac14(|B_S(G)\cup D_S(G)|)$. It is therefore natural to consider what happens when we reduce $|S|$ by one by contracting a connected subset containing two or more vertices from $S$. The result is again an outerplanar graph and we aim to show it is a smaller counterexample (unless the graph has some nice structure). Contracting an edge $uv$ can affect the degrees of the remaining vertices in the graph $G$.
Therefore $B_S(G)$ may `lose' additional vertices besides $u$ and $v$ if more than one of its neighbors are contracted and $D_S(G)$ may `lose' additional vertices if a neighbor got contracted, increasing the degree. 
We remark that vertices from $D_S(G)$ have no neighbors in $B_S(G)$, and therefore removing or contracting them does not affect the set $B_S(G)$. 

Since we minimize $|S|$ our counterexample $(G,S)$ is connected. We fix an outerplanar embedding of $G$. By definition, vertices in $D$ can have no neighbors in $B$. In fact, the following stronger claim holds.
\begin{claim}\label{cl:DS}
Every vertex $d \in D$ satisfies $\deg(d)=1$.
\end{claim}
\begin{proof}
Let $e=dv \in E(G)$ be such that $d \in D$. Suppose $d$ has a neighbour in $S$ that is not $v$. We consider the graph $G \setminus e$. Since $v \not \in S$, $S$ is a dominating set of $G \setminus e$. We find $|B_S(G \setminus e)|=|B_S(G)|$, since a vertex in $D$ has no neighbor in $B$. Similarly, $|D_S(G \setminus e)|=|D_S(G)|$.
Hence we also find that $|S|< \frac14(|B_S(G \setminus e)|+|D_S(G \setminus e)|)$. It follows that $(G \setminus e, S)$ is a counterexample to Lemma \ref{lem:outernew}. Hence since $|E(G \setminus e)| < |E(G)|$, the pair $(G \setminus e,S)$ is smaller with respect to \ref{eq:ddagger}, contradicting our choice of $(G, S)$.
%Recall that our counterexample $(G,S)$,  amongst all counterexamples $(F,M)$, 
%minimizes $|M|$; maximizes $|B_M(F) \cup D_M(F)|$; minimizes $|E(B_M(G))|$; minimizes $|V(F)|$; maximizes $|E(M,N(F))|$; and finally minimizes $|E(F)|$.
%Hence the pair $(G \setminus e, S)$ is a smaller counterexample than $(G, S)$, contradicting our choice of $(G, S)$.
\cqed
\end{proof}

\begin{claim}\label{cl:Sstable}
The set $S$ is a stable set.
\end{claim}
\begin{proof}
    Assume towards a contradiction that there are two vertices $u$ and $w$ in $S$ that are adjacent. We consider the graph $G \setminus uw$ in which $S$ is still a dominating set. Since $u,w \notin B_S(G)$, $|B_S(G \setminus e)|=|B_S(G)|$. Removing edges does not affect $D$, hence $|D_S(G \setminus e)|=|D_S(G)|$.  Hence the pair $(G \setminus e,S)$ is smaller with respect to \ref{eq:ddagger}, contradicting our choice of $(G, S)$. \cqed
\end{proof}
We are now ready to make more refined observations about the structure of $(G,S)$. When considering a pair $(G',S')$ that is smaller than $(G,S)$ with respect to \ref{eq:ddagger} with $V(G')\subseteq V(G)$, it can be useful to refer informally to vertices that belong to $B_S(G)$ but not to $B_{S'}(G')$ as \emph{lost} vertices (similarly for $D_S(G)$ and $D_{S'}(G')$). The number of lost vertices is an upper bound on $|B_S(G)\cup D_S(G)|-|B_{S'}(G') \cup D_{S'}(G')|$.

We need the following notation. Let $P \subseteq E(G)$. We denote the multigraph obtained from $G$ by contracting every edge in $P$ and deleting self-loops by $G / P$.
Note $G / P$ remains outerplanar and may contain parallel edges. 

%We remark that no vertex in $S$ has only neighbors in $D$. Indeed, if some $s \in S$ has only neighbors in $D$ then we remove $s$ and its neighborhood from the graph and we have a smaller counterexample, as we reduced $|S|$ by one and $|D|$ by at most $3$. 
In fact, we will show the following:
\begin{claim}
\label{cl:nbhoodD}
If $d\in D$ and $s\in N(d)\cap S$, then $s$ has a neighbor in $B$.
\end{claim}
\begin{proof}
Suppose that $d\in D$ is adjacent to $s\in S$. Since $d\in D$, we find that $s$ has degree at most $3$. Say $s$ has neighbors $w_1$ and $w_2$ which are not in $B$ (possibly equal, but both not equal to $d$).  By Claim \ref{cl:Sstable}, we find $w_1,w_2\not\in S$, and so $w_1,w_2$ have degree at most $3$. If $w_1 \in D$, then $deg(w_1)=1$ by Claim \ref{cl:DS}. If $w_2 \in D$, then we remove $s$ and its neighbours from the graph and found a smaller counterexample as we reduced $|S|$ by one and $|D|$ by at most $3$. If $w_2 \notin D$, it has at most $2$ neighbours in $B$. When we remove $N[s]$ from the graph, $|S|$ goes down by one and $|B \cup D|$ goes down by at most four (`counting' the two outside neighbors, $w_1$ and $d$), contradicting the minimality of our counterexample (in terms of $|S|$).
So we can assume $w_1,w_2\notin D$. Since $w_1,w_2$ are of degree at most $3$, they can each have at most $2$ neighbors in $B$.  

Recalculating now that we know that $w_1,w_2\not\in B\cup D$, if $N(\{w_1, w_2 \})$ contains at most three vertices of degree four in $B$, then $|B \cup D|$ goes down by at most four in $G \setminus N[s]$. This would be a contradiction as $(G, \setminus N[s], S \setminus \{ s \})$  would be a smaller counterexample (in terms of $|S|$), see Figure \ref{fig:4}).
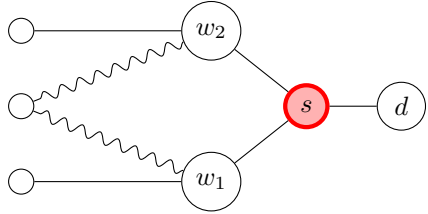
\begin{figure}[!h]
    \centering
    \begin{tikzpicture}[scale=0.5]  
    \def\height{2};
    \def\width{2.5};
    
    \node[snode] (s) at (\width, 0){$s$};
    
    \node[normal node] (w_2) at (0, \height){$w_2$};
    \node[normal node] (w_1) at (0, -\height) {$w_1$};
    \node[normal node] (d) at (2*\width, 0) {$d$};
    \node[normal node] (v_1) at (-2*\width, \height) {};
    \node[normal node] (v_2) at (-2*\width, 0) {};
    \node[normal node] (v_3) at (-2*\width, -\height) {};
    \draw (s) to (d);
    \draw (s) to (w_1);
    \draw (s) to (w_2);
    \draw (w_1) to (v_3);
    \draw (w_2) to (v_1);
    \draw[maybe edge] (w_2) to (v_2);
    \draw[maybe edge] (w_1) to (v_2);
\end{tikzpicture}
    \caption{An illustration of the case when $w_1,w_2$ are not in $B$ and together have three neighbors which are not $s,w_1$ or $w_2$. At least one of the wavy edges is present and at least one of $w_1,w_2$  has degree three in the picture. In particular, $w_1$ and $w_2$ are not adjacent. }
    \label{fig:4}
\end{figure}

Hence, $w_1,w_2$ have exactly four neighbors in $B$, all of which are of degree exactly $4$. In particular, $w_1,w_2$ are not adjacent. 

Let $G'$ be the graph obtained from $G$ by adding the edge $w_1w_2$. Considering the local rotation of the three neighbors of $s$ in an outerplanar embedding of $G$, we note that $w_1,w_2$ are consecutive neighbors of $s$. We can draw the edge $w_1w_2$ close to the path $w_1 \dd s \dd w_2$ keeping the embedding outerplanar (since $d_1$ is on the outer face still, and hence so is $s$).
It follows that $w_1, w_2 \in B_S(G') \setminus B_S(G)$, so $|B_S(G) \cup D_S(G)| < |B_S(G') \cup D_S(G')|$ (since $w_1, w_2$ have no neighbors in $D$). Thus, $(G', S)$ is a smaller counterexample (in terms of $-|B_S(G) \cup D_S(G)|$), a contradiction. 
\cqed
\end{proof}

\begin{claim}\label{cl:Snocommon}
No two vertices in $S$ have a common neighbor.
\end{claim}

\begin{proof}
Assume towards a contradiction that there are two vertices $s_1$ and $s_2$ in $S$ that have a common neighbor $v$. Since $S$ is a stable set (Claim~\ref{cl:Sstable}), we have $v\not\in S$. 

We will consider the outerplanar graph $G'=G/P$ obtained by contracting $P= \{s_1v, vs_2 \}$ into a single vertex $v_P$. Let $S'=S\setminus \{s_1, s_2\}\cup \{v_P\}$. We use the abbreviations $B'=B_{S'}(G')$ and $D'=D_{S'}(G')$. 

We will again do a case analysis, on the union of the neighbors of $s_1$ and  the neighbors of $s_2$ in $D$, to find a smaller counterexample. Note that $v \notin D$ by Claim \ref{cl:DS}.
If $|B'|+|D'|\geq |B|+|D|-4$, then $(G', S')$ is a smaller counterexample (in terms of $|S|$). Note that vertices in $B \setminus B'$ have at least two neighbors in the set $\{s_1,v,s_2\}$.
\begin{itemize}
    \item Suppose first that for some $i\in\{1,2\}$, $s_i$ is adjacent to at least two vertices in $D$. Then $v$ is the only other neighbor of $s_i$, so the graph $G''$ obtained from $G$ by deleting $s_i$ and its two neighbors in $D$, satisfies $|B(G'')\cup D(G'')|\geq |B\cup D|-3$ whereas the set $S''=S \setminus \{s_i\}$ is dominating. This gives a smaller counterexample (in terms of $|S|$).
    \item Suppose that both $s_1,s_2$ are adjacent to a single vertex in $D$. Then both have degree at most $3$. Let $d_1,d_2\in D\setminus\{v\}$ be the neighbors of $s_1,s_2$ respectively (where $d_1,d_2$ might be equal). The graph $G'$ is a smaller counterexample unless we lost two vertices from $B$ besides possibly $v$, that is, $|B'|\leq |B\setminus \{v\}|-2$. Any vertex lost from $B\setminus\{v\}$ must be adjacent to two vertices among $\{s_1,v,s_2\}$ (as otherwise its degree did not change), and since both $s_1$ and $s_2$ already have two named neighbors, $G'$ is a counterexample unless there is, for each $i\in \{1,2\}$, a common neighbor $b_i\in B$ of $s_i$ and $v$, and all named vertices are distinct.  

    Since $d_1\in D$ and $b_1\in B$, we find that $b_1d_1$ is not an edge of $G$. Since $s_1$ has three neighbors, $b_1$ and $d_1$ are consecutive neighbors and the edge $b_1d_1$ can be added without making the graph non-planar.  Consider adding the edge $b_1d_1$ in $G$ along the path $b_1s_1d_1$, such that there are no vertices in between the edge and the path. This may affect whether $s_1$ is on the outer face, but it does not affect whether $s_2$ is on the outer face. Therefore, after contracting $P$ this adjusted graph, the obtained graph $G''$ is still outerplanar. Moreover, $b_1$ has the same degree in $G''$ as in $G$, and so $|B_{S'}(G'')\cup D_{S'}(G'')|\geq |B|+|D|-4$ and  $G''$ is a smaller counterexample (in terms of $|S|$).
     \item Suppose that $s_1$ is adjacent to a vertex $d_1$ in $D$ and $s_2$ is not (the symmetric case is analogous). There can be at most three vertices in $B\setminus\{v\}$ which are adjacent to two vertices in $s_1,v,s_2$ (as only one can be adjacent to $s_1$ and $v,s_2$ have at most two common neighbors since the graph is outerplanar). The only way in which $G'$ is not a counterexample, is when there is a common neighbor $b_1$ of $s_1$ and $v$ and two common neighbors $b_2,b_3$ of $s_2$ and $v$ with all named vertices distinct. As before, we may now add the edge $b_1d_1$ in order to obtain a smaller counterexample $G''$: we add the edge $b_1d_1$ and then contract $P$, getting a smaller counterexample (in terms of $|S|$).
    \item Finally, suppose that $s_1$ and $s_2$ have no neighbors in $D$. 
    By outerplanarity, there are at most four vertices with two neighbors among $\{s_1,v,s_2\}$. Hence $G'$ is a counterexample unless there are exactly four (the only vertices `lost' from $B\cup D$ are either $v$ or among such common neighbors, since $s_1$ and $s_2$ have no neighbors in $D$). All four vertices are adjacent to $v$, because otherwise $G$ contains a $K_{2,3}$-minor\footnote{The vertices $s_1,s_2$ can have at most one further common neighbor $v^*$ besides $v$. If $v^*$ exists, we contract it with $s_1$ and $s_2$. We find a $K_{2,3}$ subgraph with $v,v^*$ on one side and the three other common neighbors on the other side.}, a contradiction. In particular, $G'$ is a counterexample unless there are two common neighbors of $v$ and $s_1$ and two common neighbors of $v$ and $s_2$ (and so $d(v)\geq 6$ and $v\in B$).
    
    Fix a clockwise order $w_1,w_2,\ldots, w_d$ on the neighbors of $v$ such that the path $w_1vw_d$ belongs to the boundary of the outer face. Let $i\neq j$ such that $w_i=s_1$ and $w_j=s_2$. After relabelling, we may assume $i<j$. Since $s_1$ and $s_2$ both have two common neighbors with $v$, we find $i>1$, $j<d$ and $i+1<j-1$.
    The vertices adjacent to multiple vertices in $\{s_1,v,s_2\}$ are $w_{i-1},w_{i+1},w_{j-1}$ and $w_{j+1}$.
    We create a new graph $G''$ by replacing $v$ with two adjacent vertices $v_1$ and $v_2$, where $v_1$ is adjacent to $w_1,w_2,\ldots, w_{i+1}$ and $v_2$ to $w_{i+2},\ldots, w_d$.
    This graph is outerplanar because both $v_1$ and $v_2$ have an edge incident to the outer face. Moreover, $d(v_1)$ and $d(v_2)$ are both at least $4$, since they are adjacent to each other, to either $s_1$ or $s_2$ and to at least two vertices among $w_1,\dots,w_d$. The set $S$ is still a dominating set, but $|B(G'')\cup D(G'')|>|B\cup D|$ so this is a smaller counterexample (since we choose our counterexample to minimize $|S|$ and with respect to that maximize $|B_S(G) \cup D_S(G)|$). 
    \end{itemize}
In all cases, we found a smaller counterexample. This contradiction proves the claim. \cqed
\end{proof}
With the claims above in hand, we now analyze the structure of $H=H_G(S)$ as described in the notation section more closely.
Note that the for each $u\in V(G)\setminus S$ the vertex $s(u)$ is uniquely defined by Claim~\ref{cl:Snocommon}.

Recall that $H$ is outerplanar. It follows that there is a vertex $s_1 \in V(H)$ with at most 2 distinct neighbors in $H$. 

We start with an easy observation.

\begin{observation}\label{obs:B}
Let $b \in B$ and $s(b)$ be its unique neighbor in $S$. Then there exists $w \in N(b) \setminus \{s(b)\}$, such that its  unique neighbor $s(w) \in S$ is not equal to $s(b)$. 
\end{observation}

Indeed, the vertex $b$ can have at most two common neighbors with $s(b)$ (otherwise there would be a $K_{2,3}$, contradicting outerplanarity), and a vertex in $B$ has degree at least $4$ by definition.

Note that the vertex $s_1$ has at least one neighbor in $H$. Indeed, if $s_1$ has no neighbor in $H$, then $N[s_1]$ is a connected component in $G$. Since $G$ is connected, $G=N[s_1]$. By Observation~\ref{obs:B}, we have $B=\emptyset$, so $|D\cup B|\leq 3$. 

\begin{claim}\label{cl:s1degree2}
The vertex $s_1$ has precisely two neighbors in $H$.
\end{claim}
\begin{proof}
Assume towards a contradiction that $s_1$ has a single neighbor $s_2$ in $H$. Let $v_1,\ldots, v_k$ be the vertices in $N[s_1]$ that have a neighbor in $N[s_2]$, and conversely let $u_1,\dots,u_\ell$ be the vertices in $N[s_2]$ that have a neighbor in $N[s_1]$. Note that by Claims~\ref{cl:Sstable} and~\ref{cl:Snocommon}, all of $\{v_1,\ldots,v_k,u_1,\ldots,u_\ell,s_1,s_2\}$ are pairwise distinct. 
If $\ell \geq 3$, then contracting the connected set $N[s_1]$ in $G$ gives a $K_{2,3}$ on the contracted vertex and $s_2$ on one side and $u_1,u_2,u_3$ on the other.
We derive that $\ell \leq 2$, and by symmetry, $k \leq 2$. By Observation~\ref{obs:B}, the only neighbors of $s_1$ that belong to $B$ are in $\{v_1,v_2\}$. As we assumed that $s_1$ has degree 1 in $H$, we have $N[v_i]\subseteq N[s_1] \cup \{u_1,u_2\}$ for $i \in \{1,2\}$. We will do a case distinction on $N[s_1]\cap D$.
\begin{itemize}
    \item If $s_1$ has no neighbor in $D$, we delete $N[s_1]$, and note that $D_S(G)=D_{S \setminus s_1}(G\setminus N[s_1])$, while $B_S(G) \setminus \{v_1,v_2,u_1,u_2\} \subseteq B_{S \setminus s_1}(G\setminus N[s_1])$. Therefore, 
    \[
    |S \setminus \{s_1\}|\geq \frac14 \cdot (|D_{S \setminus s_1}(G\setminus N[s_1])|+|B_{S \setminus s_1}(G\setminus N[s_1])|).
    \]
    So we have found a smaller counterexample (in terms of $|S|$).
    \item Suppose $s_1$ has two neighbors $d_1\neq  d_2$ in $D$. Then $v_2$ does not exist since $d(s_1) \leq 3$ and because $v_1,v_2$ are distinct from $d_1,d_2$ (vertices in $D$ have degree 1 by Claim \ref{cl:DS}). Suppose first that $v_1$ has degree at least $4$. Let $x$ be its neighbor distinct from $u_1,u_2,s_1$. By assumption on $s_1$, the vertex $x$ has no neighbor in $S \setminus \{s_1,s_2\}$. Therefore, $x$ is adjacent to $s_1$. However, $x$ is distinct from $d_1$, $d_2$ and $v_1$, which contradicts $d(s_1) \leq 3$. This case is illustrated in Figure \ref{fig:s3_does_not_exist_s1_has_2_D_nbrs}. Hence $v_1 \notin B$ and removing $N[s_1]$ now gives a smaller counterexample, a contradiction (in terms of $|S|$).
    %\linda{it's a counterexample because the only way we lose more than four vertices is if $v_1 \in D$ and $u_1, u_2 \in B$ which can't happen since $B$ is anticomplete to $D$.}
        \begin{figure}[!h]
    \centering
    \begin{tikzpicture}[scale=0.5]  
    \def\height{2};
    \def\width{2.5};
    
    \node[snode] (s_1) at (\width, 0){$s_1$};
    \node[snode] (s_2) at (4*\width, 0) {$s_2$};
    
    \node[normal node] (d_1) at (0, \height){$d_1$};
    \node[normal node] (d_2) at (0, -\height) {$d_2$};
    \node[normal node] (v_1) at (2*\width, 0) {$v_1$};
    \node[special node] (x) at (2*\width, -1*\height) {$x$};
    \node[normal node] (u_1) at (3*\width, \height) {$u_1$};
    \node[maybe node] (u_2) at (3*\width, -\height) {$u_2$};
    
    \draw (u_1) to (s_2);
    \draw (u_2) to (s_2);
    \draw (v_1) to (u_1);
    \draw (v_1) to (u_2);
    \draw (s_1) to (v_1);
    \draw[maybe edge, thin] (s_1) to (x);
    \draw[maybe edge, thin] (v_1) to (x);
    \draw (s_1) to (d_1);
    \draw (s_1) to (d_2);
    
     \path
    (u_1) edge node[above,rotate= 90] {$=$?} (u_2) [dotted];
\end{tikzpicture}
    \caption{An illustration of the case where $s_1$ has degree one in $H$ and two neighbors $d_1, d_2 \in D$ in $G$. If $v_1 \in B$, then some vertex $x$ exists such that both wavy edges are present in $G$, a contradiction.}
    \label{fig:s3_does_not_exist_s1_has_2_D_nbrs}
\end{figure}
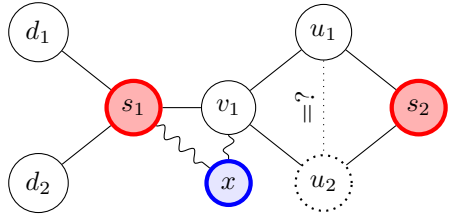
    \item Suppose that $s_1$ has a single neighbor $d_1$ in $D$. $(G \setminus N[s_1], S \setminus s_1)$ is smaller than $(G, S)$ with respect to \ref{eq:ddagger}. $(G \setminus N[s_1], S \setminus s_1)$ cannot be a counterexample. Hence all of $u_1,u_2, v_1,v_2$ exist and belong to $B$. In particular, $v_1,v_2$ both have degree at least 4. Each of $v_1$ and $v_2$ can only have neighbors within $\{s_1,v_1,v_2,u_1,u_2\}$ because a neighbor $x$ not within $\{s_1,v_1,v_2,u_1,u_2\}$ is a neighbor of $s_1$, but $d(s_1)\leq 3$. Therefore, both $v_1$ and $v_2$ are adjacent to $u_1$ and $u_2$. Together with $s_1$, this forms a $K_{2,3}$ subgraph  (see Figure \ref{fig:s3_does_not_exist_s1_has_1_D_nbr}): a contradiction. 
    \begin{figure}[!h]
        \centering
        \begin{tikzpicture}[scale=0.5]  
    \def\height{2};
    \def\width{2.5};
    
    \node[snode] (s_1) at (\width, 0){$s_1$};
    \node[snode] (s_2) at (4*\width, 0) {$s_2$};
    
    \node[normal node] (d_1) at (0, 0){$d_1$};
    \node[normal node] (v_1) at (2*\width, \height) {$v_1$};
    \node[normal node] (v_2) at (2*\width, -\height){$v_2$};
    \node[normal node] (u_1) at (3*\width, \height) {$u_1$};
    \node[normal node] (u_2) at (3*\width, -\height) {$u_2$};
    
    \draw (u_1) to (s_2);
    \draw (u_2) to (s_2);
    \draw (v_1) to (u_1);
    \draw (v_1) to (u_2);
    \draw (s_1) to (v_1);
    \draw (s_1) to (d_1);
    \draw (u_1) to (v_2);
    \draw (u_2) to (v_2);
    \draw (s_1) to (v_2);
\end{tikzpicture}
        \caption{An illustration of the case where $s_1$ has degree one in $H$ and has exactly one neighbor in $D$ in $G$. We reduce to the case in which the depicted graph is a subgraph of $G$. We find a contradiction since the depicted graph contains a $K_{2,3}$.}
        \label{fig:s3_does_not_exist_s1_has_1_D_nbr}. 
    \end{figure}
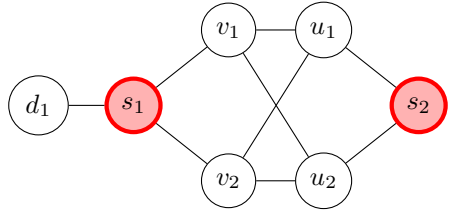%
\end{itemize}% 
\cqed
\end{proof}
So $s_1$ has two neighbors in $H$. Let $s_2, s_3 \in V(H)$ be its neighbors. We again denote $v_1,v_2$ and $u_1,u_2$ for the potential neighbors of $s_1$ and $s_2$ in $G$ corresponding to the edge $s_1s_2$ in $H$. Similarly, let $w_1,\ldots, w_p$ be the vertices in $N[s_1]$ that have a neighbor in $N[s_3]$ (within $G$), and conversely let $x_1,\dots,x_q$ be the vertices in $N[s_3]$ that have a neighbor in $N[s_1]$. By the same argument as before for $s_1$ and $s_2$, we obtain $p \leq 2$ and $q \leq 2$ and that all of $\{w_1,w_2,x_1,x_2,s_1,s_3\}$ are pairwise distinct. However, there may be a vertex in $\{w_1,w_2\}\cap \{v_1,v_2\}$; there may not be two such vertices since this would lead to a $K_{2,3}$-minor (with vertices $\{v_1,w_1\}$ and $\{v_2,w_2\}$ in one part, and $s_1,\{s_2,u_1,u_2\}, \{s_3,x_1,x_2\}$ in the other). 

Our general approach is to delete $N[s_1]$ and add edges between $\{u_1,u_2\}$ and $\{x_1,x_2\}$ as appropriate so as to mitigate the impact on $|B \cup D|$. If this does not work, we obtain further structure on the graph which we exploit to create a different smaller counterexample. We will repeatedly apply the following Observation~\ref{obs:flipping}. Sometimes when deleting vertices and edges from the graph $G$, the result is a disconnected graph, so we can perform the "flipping" operation described below, and connect the different components to get a smaller counterexample $(G',S')$.
\begin{observation}[Flipping]
\label{obs:flipping}
Let $G$ be the disjoint union of two outerplanar graphs $O_1$ and $O_2$.
Consider an outerplanar embedding of $G$, and let $(u_1,u_2,\ldots,u_q)$ denote the outer face of $G[O_2]$ in clockwise order. We can obtain a different outerplanar embedding of $G$ by reversing the order of $O_2$ without modifying the embedding of $O_1$, so that the outer face of $G[O_2]$ is $(u_q,\ldots,u_2,u_1)$ in clockwise order.
\end{observation}

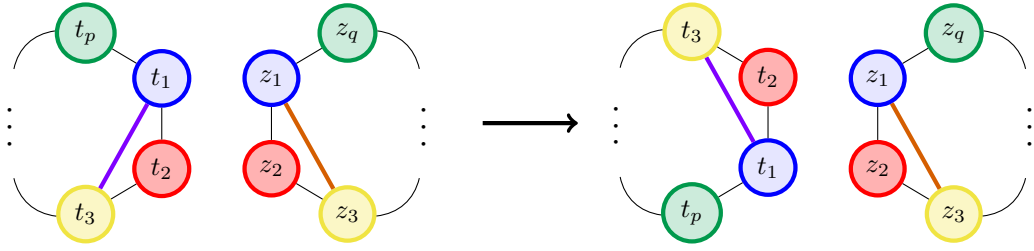
\begin{figure}[]
    \centering
    \begin{tikzpicture}[scale=0.4]
\tikzstyle{one} = [special node];
\tikzstyle{three} = [shape= circle, draw=gelb, ultra thick, fill=gelb!30];
\tikzstyle{two} = [snode]
\tikzstyle{four} =  [shape= circle, draw=grun, ultra thick, fill=grun!20]

\tikzstyle{edge style one} = [lila, ultra thick]
\tikzstyle{edge style two} = [darkorange, ultra thick]
\def\height{1.5}
\def \secondheight{3}
\def\width{2.5}
\def \secondwidth{2*\width}

\node(before) at (0,0){
\begin{tikzpicture}[scale = 0.08]
\node (T_before) at (0,0) {
    \begin{tikzpicture}[scale = 0.4]
        \node[one] (t_1) at (0, \height) {$t_1$};
        \node[two] (t_2) at (0, -\height) {$t_2$};
        \node[three] (t_3) at (-\width, -\secondheight){$t_3$};
        \node[four] (t_p) at (-\width, \secondheight) {$t_p$};
        \draw (t_1) to (t_2);
        \draw (t_1) to (t_p);
        \draw (t_2) to (t_3);
        \draw[edge style one] (t_1) to (t_3);
        %need a partial arc between t_1 and t_p
        \node (tp') at (-\secondwidth, \height){};
        \node (t3') at (-\secondwidth, -\height){};
        
        \draw[bend left=35] (t_3) to (t3');
        \draw[bend right = 35] (t_p) to (tp');
        
        \node[rotate =90] at (-\secondwidth, 0) {\Large $\dots$};
    \end{tikzpicture}
};
\node (Z_before at (T_before.east)[xshift=3.2cm]{
    \begin{tikzpicture}[scale =0.4]
        \node[one] (z_1) at (0, \height) {$z_1$};
        \node[two] (z_2) at (0, -\height) {$z_2$};
        \node[three] (z_3) at (\width, -\secondheight){$z_3$};
        \node[four] (z_p) at (\width, \secondheight) {$z_q$};
        \draw (z_1) to (z_2);
        \draw (z_1) to (z_p);
        \draw (z_2) to (z_3);
        \draw[edge style two] (z_1) to (z_3);
        %need a partial arc between t_1 and t_p
        \node (before_zp') at (\secondwidth, \height){};
        \node (z3') at (\secondwidth, -\height){};
        
        \draw[bend right=35] (z_3) to (z3');
        \draw[bend left = 35] (z_p) to (before_zp');
        
        \node[rotate =90] at (\secondwidth, 0) {\Large $\dots$};
    \end{tikzpicture}
};
\end{tikzpicture}
};
\node (after) at (before.east)[xshift = 4cm]{
    \begin{tikzpicture}
        \node (T_after) at (0,0) {
            \begin{tikzpicture}[scale = 0.4]
                \def\height{-1.5}
                \def\secondheight{-3}
                \node[one] (t_1) at (0, \height) {$t_1$};
                \node[two] (t_2) at (0, -\height) {$t_2$};
                \node[three] (t_3) at (-\width, -\secondheight){$t_3$};
                \node[four] (t_p) at (-\width, \secondheight) {$t_p$};
                \draw (t_1) to (t_2);
                 \draw (t_1) to (t_p);
                \draw (t_2) to (t_3);
                 \draw[edge style one] (t_1) to (t_3);
                 %need a partial arc between t_1 and t_p
                \node (after_tp') at (-\secondwidth, \height){};
                \node (t3') at (-\secondwidth, -\height){};
        
                \draw[bend right=35] (t_3) to (t3');
                \draw[bend left = 35] (t_p) to (after_tp');
        
                \node(arrow_end)[rotate =90] at (-\secondwidth, 0) {\Large $\dots$};
            \end{tikzpicture}
            };
        \node (Z_after at (T_after.east)[xshift=3.2cm]{
        \begin{tikzpicture}[scale =0.4]
            \node[one] (z_1) at (0, \height) {$z_1$};
            \node[two] (z_2) at (0, -\height) {$z_2$};
            \node[three] (z_3) at (\width, -\secondheight){$z_3$};
            \node[four] (z_p) at (\width, \secondheight) {$z_q$};
            \draw (z_1) to (z_2);
            \draw (z_1) to (z_p);
            \draw (z_2) to (z_3);
            \draw[edge style two] (z_1) to (z_3);
            %need a partial arc between t_1 and t_p
            \node (zp') at (\secondwidth, \height){};
            \node (z3') at (\secondwidth, -\height){};
        
            \draw[bend right=35] (z_3) to (z3');
            \draw[bend left = 35] (z_p) to (zp');
        
            \node[rotate =90] at (\secondwidth, 0) {\Large $\dots$};
    \end{tikzpicture}
};
\node (x) at (-3,0) {};
\node (y) at (-1.5, 0) {};
\draw[->, ultra thick] (x) to (y);

\end{tikzpicture}
};
\end{tikzpicture}
    \caption{An illustration of Observation~\ref{obs:flipping}.}
    \label{fig:flip}
\end{figure}

An example of the observation above is given in Figure \ref{fig:flip}. Beside Observations~\ref{obs:B} and~\ref{obs:flipping}, the third useful observation is as follows.

\begin{claim}\label{cl:contained}
$N[v_1,v_2,w_1,w_2]\subseteq N[s_1]\cup \{u_1,u_2,x_1,x_2\}$. Additionally, if $\{v_1,v_2\}\cap\{w_1,w_2\}=\emptyset$, then $N[v_1,v_2]\subseteq N[s_1]\cup \{u_1,u_2\}$ and $N[w_1,w_2]\subseteq N[s_1]\cup \{x_1,x_2\}$.
\end{claim}
\textcolor{blue}{Alex: I do not really understand the second property of this claim, is that not part of the definition? What we are proving here seems the same as Claim 2.12. }
% This observation is argued similarly to Observation~\ref{obs:B}. 
\begin{proof}
The first statement holds by the assumption that the only neighbors of $s_1$ in $H$ are $s_2$ and $s_3$. To see the second statement, suppose that $v_1$ is adjacent to $x_1$ (the other cases are similar). Then there is a $K_{2,3}$-minor if $w_2$ exists: $\{x_1,x_2,s_3\}$ and $\{s_1\}$ are both adjacent to $v_1,w_1$ and $w_2$.

So it must be the case that $w_2$ does not exist. We first consider the case in which $s_1$ has a neighbor in $D$. Since $s_1$ can have at most 3 neighbors, we find that $v_2$ does not exist. We can remove $N[s_1]$ and add edges $x_1u_1,x_1u_2,x_2u_2$ in order to obtain a smaller counterexample (in terms of $|S|$). 
(Indeed, at most one of the two edges $v_1x_1$ and $w_1x_2$ exists in $G$ so $v_1$ and $w_1$ cannot both be in $B$. So deleting $N[s_1]$ decreases $B\cup D$ by at most $4$.)

If $s_1$ has no neighbors in $D$, then we may remove $N[s_1]$ and add edges between $\{x_1,x_2,u_1,u_2\}$ in such a way that at most one of the vertices in that set decreases in degree except if $u_2$,$x_2$ both do not exist, in which case at most one of $v_1,w_1$ is in $B$; this again gives a smaller counterexample (in terms of $|S|$).
\end{proof}
Since $s_1$ is adjacent to $s_2$ and $s_3$ in $H$, all of $u_1,x_1,v_1$ and $w_1$ exist. We assume that either $\{v_1,v_2\}\cap \{w_1,w_2\}=\emptyset$ or $v_1=w_1$. Note that $\{u_1,u_2\}\cap \{x_1,x_2\}=\emptyset$ since $s_2$ and $s_3$ do not have common neighbors by Claim~\ref{cl:Snocommon}. 
See Figure \ref{fig:s2s3_summary} for an illustration. For simplicity, when depicting which edges to add in which cases, we represent ``$u_2$ does not exist'' as ``$u_2$ is possibly equal to $u_1$'' (and variations). This means merely that if $u_2$ does not exist then the 
edges involving $u_2$ involve $u_1$ instead -- multiple edges are ignored.

\begin{figure}[!h]
    \centering
    %\section{everything figure style one}

\begin{tikzpicture}[scale=0.7] 
\def\height{1.5};
\def\width{2.5};

\node[snode] (s_2) at (0, 0){$s_2$};
\node[snode] (s_1) at (3*\width ,0){$s_1$};
\node[snode] (s_3) at (6* \width, 0){$s_3$};

\node[normal node] (v_1) at (2*\width, \height){$v_1$};
\node[normal node] (w_1) at (4*\width, \height){$w_1$};
\node[maybe node] (v_2) at (2*\width, -\height) {$v_2$};
\node[maybe node] (w_2) at (4*\width, -\height){$w_2$};

\node[normal node] (u_1) at (1*\width, \height){$u_1$};
\node[maybe node] (u_2) at (1*\width, -\height){$u_2$};

\node[normal node] (x_1) at (5*\width, \height){$x_1$};
\node[maybe node] (x_2) at (5*\width, -1*\height){$x_2$};

%edges
\draw (s_1) to (v_1);
\draw (s_1) to (w_1);
\draw (v_1) to (u_1);
\draw (v_2) to (u_2);
\draw (s_2) to (u_1);
\draw (s_2) to (u_2);
\draw (s_3) to (x_1);
\draw (s_3) to (x_2);
\draw (w_1) to (x_1);
\draw (w_2) to (x_2);
\draw (s_1) to (v_2);
\draw (s_1) to (w_2);

%u_1 and u_2 are maybe equal, maybe adjacent
\path
(u_1) edge node[above,rotate= 90] {$=$?} (u_2) [maybe equal]
(x_1) edge node[above,rotate= 90] {$=$?} (x_2) [maybe equal]
(v_1) edge node[above,rotate= 90] {$=$?} (v_2) [maybe equal]
(w_1) edge node[above,rotate= 90] {$=$?} (w_2) [maybe equal]
%v_1 = w_1?
(v_1) edge node[above] {$=$?} (w_1) [maybe equal];
\end{tikzpicture}
    \caption{When $s_1$ has exactly two neighbors $s_2, s_3$ in $H$, each of $s_2, s_3$ has at most two neighbors with edges to vertices in $N[s_1]$. Moreover, $s_2$ and $s_3$ may have at most one common neighbor in $N[s_1]$. We draw vertices which may not exist in $G$ as a dotted circle and connect vertices which may be equal with dotted edges. There may be more edges present in that are not drawn. }
    \label{fig:s2s3_summary}
\end{figure}
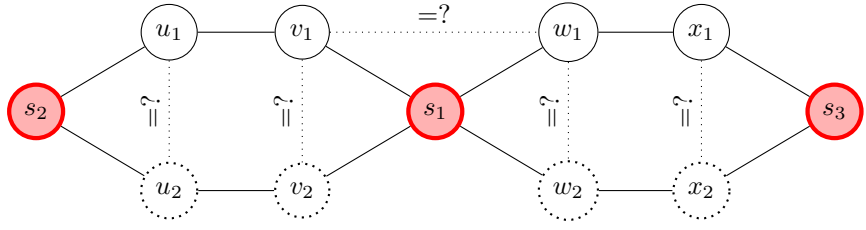

\begin{claim}
One of $w_2$ and $v_2$ exists.
\label{cl:oneofw2v2exists}
\end{claim}
\begin{proof}
Suppose that neither $w_2$ nor $v_2$ exists. It is possible that $v_1=w_1$, and that $u_2$ or $x_2$ do not exist. 
By Claim \ref{cl:contained}, if $v_1\neq w_1$, then $u_1,u_2$ are not adjacent to $w_1$ and $x_1,x_2$ are not adjacent to $v_1$.
\begin{figure}[!h]
    \centering

    \begin{tikzpicture}[scale=0.5]     %\label{Carla1_Before}
\node (before) at (0,0){
    \begin{tikzpicture}[scale=0.5]  
    \def\height{1.5};
    \def\width{2};
    
    \node[snode] (s_2) at (0, 0){$s_2$};
    \node[snode] (s_1) at (3*\width ,0){$s_1$};
    \node[snode] (s_3) at (6* \width, 0){$s_3$};
    
    \node[normal node] (v_1) at (2*\width, 0){$v_1$};
    \node[normal node] (w_1) at (4*\width, 0){$w_1$};
    
    \node[normal node] (u_1) at (1*\width, \height){$u_1$};
    \node[maybe node] (u_2) at (1*\width, -\height){$u_2$};
    
    \node[normal node] (x_1) at (5*\width, \height){$x_1$};
    \node[normal node] (x_2) at (5*\width, -1*\height){$x_2$};
    
    %edges
    \draw (s_1) to (v_1);
    \draw (s_1) to (w_1);
    \draw (v_1) to (u_1);
    \draw (v_1) to (u_2);
    \draw (s_2) to (u_1);
    \draw (s_2) to (u_2);
    \draw (s_3) to (x_1);
    \draw (s_3) to (x_2);
    \draw (w_1) to (x_1);
    \draw (w_1) to (x_2);
    \draw[maybe edge] (u_1) to (x_1);
    % %u_1 and u_2 are maybe equal, maybe adjacent
     \path
     (u_1) edge node[above,rotate= 90] {$=$?} (u_2) [dotted];
     \draw[dotted, bend right=60] (v_1) to (w_1);
     \node at (3*\width, -1.14*\height){$=$?}; 
    % \node[yellow node, dotted] (d_1) at (3*\width, -2*\height){$d_1$};
    % \node[yellow node, dotted] (d_2) at (4*\width, -2*\height){$d_2$};
    % \draw[maybe edge] (d_1) to (s_1);
     %\draw[maybe edge] (d_2) to (s_1);
    \end{tikzpicture}
    
};
\node(after) at (before.east)[anchor=east,xshift=6cm]{
    \begin{tikzpicture}[scale=.5]
    \def\height{1.5};
    \def\width{2.5};
    
    \node[snode] (s_2) at (0, 0){$s_2$};
    \node[snode] (s_3) at (3* \width, 0){$s_3$};

    \node[normal node] (u_1) at (1*\width, \height){$u_1$};
    \node[normal node] (u_2) at (1*\width, -\height){$u_2$};
    
    \node[normal node] (x_1) at (2*\width, \height){$x_1$};
    \node[normal node] (x_2) at (2*\width, -1*\height){$x_2$};
    
    %edges
    \draw (s_2) to (u_1);
    \draw (s_2) to (u_2);
    \draw (s_3) to (x_1);
    \draw (s_3) to (x_2);
    \draw[new edge] (u_1) to (x_1);
    \draw[new edge] (u_2) to (x_2);
    \draw[new edge] (u_1) to (x_2);

    %u_1 and u_2 are maybe equal, maybe adjacent
    \path
    (u_1) edge node[above,rotate= 90] {$=$?} (u_2) [dotted]
    (x_1) edge node[above,rotate= 90] {$=$?} (x_2) [dotted];
\end{tikzpicture}
};
\draw [->, ultra thick] (before)--(after);
\end{tikzpicture}

    \caption{The case where $w_2, v_2$ do not exist. The original graph is drawn at the left and the modified graph is drawn at the right. The wavy line indicates there may be an edge between $u_1$ and $x_1$. Edges that may have been added are drawn in blue. Note that $u_2$ may not exist. There may be more edges which are not drawn (for instance $v_1$ might be adjacent to $w_1$) but these edges are not relevant to our argument. }
    \label{fig:carla1}
\end{figure}
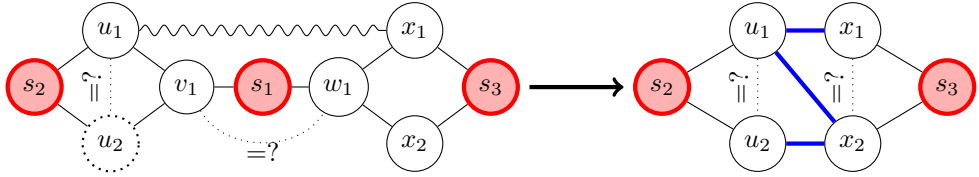

The degrees of $x_1,x_2,u_1,u_2$ in $G \setminus N[s_1]$ are at least one less than their degrees in $G$. Every vertex in $V(G) \setminus (N[s_1]\cup\{x_1,x_2,u_1,u_2\})$ has the same degree in $G$ and in $G \setminus N[s_1]$. Let  $S'=S\setminus \{s_1\}$, and note that $S'$ dominates $G \setminus N[s_1]$.
\begin{itemize}
    \item Suppose $x_2,u_2$ do not exist. If $v_1$ belongs to $B$, then it needs to have a neighbor which is not $u_1,s_1$ or one of $x_1,w_1$ (depending on whether $v_1=w_1$), so it shares a neighbor with $s_1$ which is not in $B \cup D$. This implies that if $v_1 \in B$, then $s_1$ can have a neighbor $d_1\in D$ or $w_1 \neq v_1$, but not both at the same time. It follows that regardless of whether $v_1\in B$, we have $|N[s_1]\cap(B \cup D)|\leq  2$. But now $|(N[s_1]\cup\{x_1,u_1\})\cap(B \cup D)|\leq 4$, so $(G\setminus N[s_1],S')$ is a smaller counterexample, a contradiction (in terms of $|S|$).
    \item By symmetry, we assume that $x_2$ exists. If $u_1$ and $u_2$ both exist, then they are not distinguishable at this point, which means we can swap their label. The same holds for $x_1$ and $x_2$. Hence we may assume that the vertices appear in the outer face in the order $x_1,x_2,u_2,u_1$, and that either $u_1x_1$ is an edge of $G$ or there is no edge between $\{u_1,u_2\}$ and $\{x_1,x_2\}$.
Let $G'$ be the graph obtained from $G\setminus N[s_1]$ by adding the edges $u_1x_1$ (if it is not already present), $u_1x_2$ and (if $u_2$ exists) the edge $u_2x_2$ (see Figure \ref{fig:carla1}). Note that $G'$ is outerplanar and that $S'$ dominates $G'$. Since $G$ is outerplanar, if $u_1x_1$ is an edge in $G$, then neither $u_1x_2$ nor $u_2x_2$ is an edge in $G$. In $G'$, the degrees of the vertices $u_1,u_2,x_2$ are at least as large as their respective degrees in $G$ (the degree of $x_1$ might have dropped if the edge $u_1x_1$ was already present in $G$).  Note that $|\{v_1,w_1,x_1\} \cup (N[s_1]\cap D)|\leq 4$, hence $|B_{S'}(G') \cup D_{S'}(G')|\geq |B \cup D|-4$ so $(G', S')$ is a counterexample. But $|S'| < |S|$, contradicting our choice of $(G, S)$. 
\end{itemize}
\cqed
\end{proof}
\begin{claim}\label{cl:ifv1w1thenw2exists}
If $w_1=v_1$, then $v_2$ and $w_2$ exist.
\end{claim}
\begin{proof}
By Claim~\ref{cl:oneofw2v2exists} we can assume $v_2$ exists. Suppose $w_2$ does not exist and $w_1=v_1$. We remove $N[s_1]$ and add edges between $\{u_1,u_2\}$ and $\{x_1,x_2\}$ as above to ensure that for all but at most one of them, the degree does not decrease. To see an illustration of how the edges are added, see Figure \ref{fig:carla2}.
We suppose first that there are no edges between $\{u_1,u_2\}$ and $\{x_1,x_2\}$.
The edges remedy the degree for $x_1,x_2$, since they only lost $w_1$, and for one of $u_1,u_2$ (if they both exist); indeed, it is not possible that both $u_1$ and $u_2$ are adjacent to both $v_1$ and $v_2$ (since we would obtain a $K_{2,3}$ when considering $s_1$ as well). 
By Claim~\ref{cl:contained}, the degrees of other vertices are not affected by removing $N[s_1]$. 
\begin{figure}[!h]
    \centering
    \begin{tikzpicture}[scale=0.5]     %\label{Carla1_Before}
\node (before) at (0,0){
    \begin{tikzpicture}[scale=0.5]  
    \def\height{1.5};
    \def\width{2};
    
    \node[snode] (s_2) at (0, 0){$s_2$};
    \node[snode] (s_1) at (3*\width ,0){$s_1$};
    \node[snode] (s_3) at (6* \width, 0){$s_3$};
    
    \node[normal node] (v_1) at (3*\width, 2*\height) {$v_1,w_1$};
    \node[normal node] (v_2) at (2*\width, 0){$v_2$};
    
    \node[normal node] (u_1) at (1*\width, \height){$u_1$};
    \node[maybe node] (u_2) at (1*\width, -\height){$u_2$};
    
    \node[normal node] (x_1) at (5*\width, \height){$x_1$};
    \node[maybe node] (x_2) at (5*\width, -1*\height){$x_2$};
    
    %edges
    \draw (s_1) to (v_1);
    \draw (s_1) to (v_2);
    \draw (v_1) to (u_1);
    \draw (v_2) to (u_2);
    \draw (s_2) to (u_1);
    \draw (s_2) to (u_2);
    \draw (s_3) to (x_1);
    \draw (s_3) to (x_2);
    \draw (v_1) to (x_1);
    \draw (v_1) to (x_2);
    
    %u_1 and u_2 are maybe equal, maybe adjacent
    \path
    (u_1) edge node[above,rotate= 90] {$=$?} (u_2) [dotted]
    (x_1) edge node[above,rotate= 90] {$=$?} (x_2) [dotted];
    \end{tikzpicture}
};
\node(after) at (before.east)[anchor=east,xshift=6cm]{
    \begin{tikzpicture}[scale=.5]
    \def\height{1.5};
    \def\width{2.5};
    
    \node[snode] (s_2) at (0, 0){$s_2$};
    \node[snode] (s_3) at (3* \width, 0){$s_3$};

    \node[normal node] (u_1) at (1*\width, \height){$u_1$};
    \node[normal node] (u_2) at (1*\width, -\height){$u_2$};
    
    \node[normal node] (x_1) at (2*\width, \height){$x_1$};
    \node[normal node] (x_2) at (2*\width, -1*\height){$x_2$};
    
    %edges
    \draw (s_2) to (u_1);
    \draw (s_2) to (u_2);
    \draw (s_3) to (x_1);
    \draw (s_3) to (x_2);
    \draw[new edge] (u_1) to (x_1);
    \draw[new edge] (u_2) to (x_2);
    \draw[new edge] (u_1) to (x_2);

    %u_1 and u_2 are maybe equal, maybe adjacent
    \path
    (u_1) edge node[above,rotate= 90] {$=$?} (u_2) [dotted]
    (x_1) edge node[above,rotate= 90] {$=$?} (x_2) [dotted];
\end{tikzpicture}
};
\draw [->, ultra thick] (before)--(after);
\end{tikzpicture}
    \caption{An example of the reduction for the case where $v_1$ is equal to $w_1$, $v_2$ exists and $w_2$ does not exist. We only draw edges which are relevant to our argument.}
    \label{fig:carla2}
\end{figure}
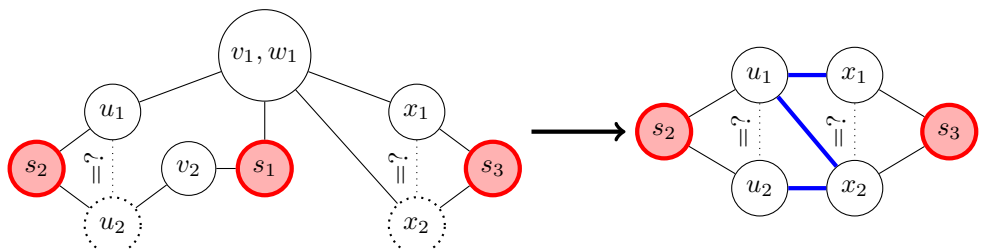
Again, since $|N[s_1]\cap (B\cup D)|\leq 3$, we have removed a vertex from $S$ and at most 4 from $B\cup D$ so we have constructed a smaller counterexample (in terms of $|S|$).

We now assume $u_1x_1$ is an edge.
\begin{itemize}
    \item Assume that $u_2$ does not exist. Now $v_2$ has degree at most $3$ unless it has a common neighbor with $s_1$, but then $s_1$ has no neighbor in $D$ and we lose only two vertices from $N[s_1]$ and possibly $u_1,x_1$.
    \item Assume now that $u_2$ exists. Both $x_1$ and $x_2$ lose at most one edge, and we can ensure both gain at least one edge. So if we lose only two vertices from $N[s_1]$, plus possibly $u_1,u_2$, then we lose at most four vertices from $B\cup D$ in total. If there are three vertices from $B\cup D$ in $N[s_1]$, then $v_2$ is adjacent to both $u_1,u_2$ and $s_1$ has a neighbor $d$ in $D$. We know that $v_1$ is adjacent to $u_1$ or $u_2$, but since there is the path $v_1x_1u_1$, we know that if $v_1$ would be adjacent to $u_2$, then there would be a $K_{2,3}$-minor with $v_1,v_2$ in one part and $\{u_1,x_1\},u_2,s_1$ in the other part.
    Since $u_2$ is not adjacent to $v_1$, we know that $u_2$ loses at most one edge when deleting $N[s_1]$ and gains an edge when we add the edge $u_2x_2$. This means we lose at most four vertices from $B\cup D$ (counting $u_1,v_1,v_2$ and $d$).
    Hence, $(G \setminus s_1, S \setminus s_1)$ is a smaller counterexample in terms of $|S|$, a contradiction.
\end{itemize}
\cqed
\end{proof}
We henceforth assume that $|\{v_1,v_2,w_1,w_2\}|\geq 3$. In particular, we may assume that $s_1$ has no neighbor in $D$.

\begin{claim}\label{cl:v1neqw1}
We have $v_1 \neq w_1$.
\end{claim}
\begin{proof}
If not, then $v_1=w_1$. By Claim~\ref{cl:ifv1w1thenw2exists}, both $v_2$ and $w_2$ exist. Let $G'$ be the outerplanar graph obtained from $G$ by splitting the vertex $v_1$ into two vertices $v_1'$ and $w_1'$, both adjacent to $s_1$ and adjacent to each other, where $v_1'$ is adjacent to $N[v_1]\cap N[s_2]$ and $w_1'$ is adjacent to $N[v_1]\cap N[s_3]$. This gives three neighbors for both $v_1'$ and $w_1'$. Since we can always add the edges $v_1'v_2$ and $w_1'w_2$ (which are chords of a cycle, using also that $N[s_1]\cap(N[s_2]\cup N[s_3])=\emptyset$), we find $|B_{S}(G')|>|B_{S}(G)|$, whereas $D_{S}(G')=D_S(G)$, $S$ is still dominating and $G'$ is outerplanar.
Since we chose $(G, S)$ to be a counterexample minimizing $|S|$ and with respect to that maximizing $|B_S(G) \cup D_S(G)|$, we reach a contradiction.
\cqed
\end{proof}
We henceforth assume that $v_1,w_1$ exist and are distinct and at least one of $v_2,w_2$ exists.
\begin{claim}
The vertices $v_2,w_2$ exist.
\label{cl:w2exists}
\end{claim}
\begin{proof}
   By Claim~\ref{cl:v1neqw1} we can assume $v_1 \neq w_1$ and by Claim~\ref{cl:oneofw2v2exists} we can assume $v_2$ exists. Suppose that $w_2$ does not exist.
    As in the previous case, $N[s_1]$ forms a vertex cut separating the component containing $N[s_2]$ from the component containing $N[s_3]$: if there was a path between $N[s_2]$ to $N[s_3]$ disjoint from $N[s_1]$, then we obtain a $K_{2,3}$-minor on vertex sets $N[s_2],\{s_1\}$ on one side and $\{v_1\}$, $\{v_2\}$, $N[s_3]$ on the other. When we delete $N[s_1]$, at most one of $u_1,u_2$ loses two neighbors, and the other (if it exists) loses only a single neighbor. The vertices $x_1,x_2$ can lose only a single neighbor. By Observation \ref{obs:flipping}, after deleting $N[s_1]$ we can align the components of $N[s_2]$ and $N[s_3]$ in such a way that we can add the edges from $\{u_1x_1,u_2x_1,u_2x_2\}$ to obtain $G'$. (For brevity, we handle the cases in which some of $u_2,x_2$ do not exist here as well, in which case we might add less edges.) After this process, we lost at most $4$ vertices of $B\cup D$, namely at most $v_1,v_2,w_1$ and one of the $u_i$ (if one of them lost two neighbors).
    Hence $(G', S \setminus s_1)$ is a counterexample and it is smaller than $(G, S)$ in terms of $|S|$, a contradiction.
    \cqed
\end{proof}

We have one final case in which $v_1,v_2,w_1,w_2$ all exist and are all distinct. Note that, as in Claim \ref{cl:w2exists} above, $N[s_1]$ forms a vertex cut separating $N[s_2]$ from $N[s_3]$. We break this problem into three subcases: The case where $x_2, u_2$ both exists, the case where exactly one of $x_2$ and $u_2$ exists and the case where neither $x_2$ nor $u_2$ exists. Since the details are not particularly illuminating, we will omit them for brevity. Appendix \ref{app:case-analysis} contains the full details of this case analysis. 
\begin{itemize}
        \item Suppose that $x_2$ and $u_2$ both exist. Note that at most one of $u_1,u_2$ and one of $x_1,x_2$ is adjacent to two vertices in $N[s_1]$. Let us assume without loss of generality that $x_1,u_2$ have at most one neighbor in $N[s_1]$. By Observation \ref{obs:flipping}, after deleting $N[s_1]$, we can re-embed the graph in a way that we can add the edges $u_1x_1,u_1x_2,u_2x_2$. This gives a smaller counterexample, since we have `fixed' the degrees of $u_1,u_2,x_1,x_2$ and only lost $N[s_1]\cap (B\cup D)$, which has size at most 4. Hence we obtain a smaller counterexample in terms of $|S|$, a contradiction.
        \item Suppose now that only $u_2$ exists. (The case in which only $x_2$ exists is analogous.)
        In this case we complete the proof by considering separately the sub-case where $s_2$ has a neighbor in $D$ and the sub-case where it does not. In both of these two sub-cases, we are able to prove there is a smaller counterexample (in terms of $|S|$) purely by analyzing the the possible edges in the graph. The details are included in Appendix \ref{app:only-one-of-x2-u2_exists}.
        \item Finally, we consider the case in which neither $u_2$ nor $x_2$ exists. If both $u_1$ and $x_1$ do not have degree exactly four, then we can remove $N[s_1]$ and add the edge $u_1x_1$; in this case we only lose a subset of $\{v_1,v_2,w_1,w_2\}$ from $B\cup D$. Hence we can assume by symmetry that $u_1$ has degree exactly four. We break this case into subcases where $s_2$ has no neighbor in $D$, where $s_2$ has only neighbors in $D \cup \{u_1 \}$ and where $s_2$ exactly one neighbor in $D$. In each case we are able to obtain a smaller counterexample (in terms of $|S|$) using elementary techniques. The details are included in Appendix \ref{app:neither-u2-nor-x2-exists}.
\end{itemize}

Above claims show that there can not be a counterexample to Lemma~\ref{lem:outernew}, which proves that Algorithm \ref{alg:1} computes a $5$-approximation of Minimum Dominating Set for outerplanar graphs.
\end{proof}
\section{Lower bound for outerplanar graphs}
\label{sec:lowerbound}
In this section we show that there is no deterministic local algorithm that finds a $(5 - \epsilon)$-approximation of a minimum dominating set on outerplanar graphs using $T$ rounds, for any $\epsilon>0$ and $T\in \mathbb{N}$. To do so we use a result from Czygrinow, Ha{\'n}{\'c}kowiak and Wawrzyniak \cite[pp.~87--88]{czygrinow2008fast} who gave a lower bound in the planar case. For $n\equiv 0 \mod 10$, they consider a graph $G_n$, which is a cycle $v_1,v_2,\dots,v_n,v_1$ where edges between vertices of distance two are added. They showed that for every local distributed algorithm $\mathcal{A}$ and every $\delta>0$ and $n_0\in \mathbb{N}$ there exists $n\geq n_0$ for which the algorithm $\mathcal{A}$ outputs a dominating set for $G_n$ that is not within a factor of $5-\delta$ of the optimal dominating set for $G_n$. Their graph $G_n$ is not outerplanar, but we can delete three of its edges to get an outerplanar graph $G^-_n$. The graph $G^-_n$ is a path $v_1 \dots v_n$ where all edges between vertices of distance two are added as in Figure \ref{fig:LowerBoundGraph}.\\

The argument of \cite{czygrinow2008fast} builds on a lower bound for local algorithms computing a maximum independent set, which in turn depends on multiple applications of Ramsey's theorem. A similar approach is used by \cite{hilke2014brief} to obtain the best-known lower-bound for planar graphs. Using the graph $G_n^-$, this approach can also be used to prove our result; the main idea is that since in the middle all the vertices `look the same', no local algorithm can do better than selecting almost all of them.

Alternatively, we can exploit the result of \cite{czygrinow2008fast} as follows. For any bound $T\in \mathbb{N}$ on the number of rounds, any vertex in $M=\{v_{2T+1},\dots,v_{n-2T-1}\}$ 
has the same local neighborhood in $G_n$ as in $G_n^-$. Since $G_n$ is rotation symmetric, a potential local algorithm also finds a dominating set $D$ for $G_n$ (for $n\geq 4T+2$), and with the result of \cite{czygrinow2008fast} we obtain $|D|\geq (5-\delta)\gamma(G_n)$. For $n$ sufficiently large with respect to $T$, the set $D$ is the same as the set $D'$ that the algorithm would give for $G_n$ up to at most $\delta n/10$. Since $n\equiv 0 \mod 10$, $\gamma(G_n)=\gamma(G_n^-)=\frac{n}5$ and we find the desired lower-bound $|D'|\geq \left(5-\frac{\delta}{2}\right)\gamma(G_n^-)\geq (5-\epsilon)\gamma(G_n^-) $ for $\delta$ small enough.

\section{Conclusion}

Through a rather intricate analysis of the structure of outerplanar graphs, we were able to determine that a very naive algorithm gives a tight approximation for minimum dominating set in outerplanar graphs in $O(1)$ rounds. While there are some highly non-trivial obstacles to extending such work to planar graphs, we believe that similar techniques can be used to vastly improve the state of the art for triangle-free planar graphs and for $C_4$-free planar graphs. In the first case, recall that a $32$-approximation is known~\cite{alipour2020distributed}, and there is a simple construction (a large $4$-regular grid) showing that $5$ is a lower bound. We believe that $5$ is the right answer. In the second case, an $18$-approximation is known~\cite{alipour2020local}, and there is no non-trivial lower bound. We refrain from conjecturing the right bound here -- we simply point out that there is no reason yet to think $3$ is out of reach. We believe that very similar techniques to the ones developed here can be used to obtain a $9$-approximation, and possibly lower.

\paragraph{Acknowledgments}
We thank the referees for helpful comments which improved the presentation of the paper.
\bibliography{bibliography_LocalDominating}
\newpage
\begin{appendices}
\section{The case when $v_1, v_2, w_1, w_2$ all exist and are all distinct}
\label{app:case-analysis}
This appendix will provide the details of the case analysis from end of the proof of Lemma \ref{lem:outernew}. Suppose $v_1, v_2, w_1,w_2$ all exist and all are distinct.
\subsection{The case when exactly one of $x_2, u_2$ exists}
\label{app:only-one-of-x2-u2_exists}
Suppose now that only $u_2$ exists. (The case in which only $x_2$ exists is analogous.) 
        Note that $s_2$ can have at most one neighbor in $D$. See also Figure \ref{fig:c100}.
        \begin{itemize}
    \item Suppose first that $s_2$ has a neighbor $d \in D$. We delete and add edges (if needed) and renumber such that $u_1$ is adjacent to $v_1$, $u_2$ to $v_2$ and $u_1$ to $u_2$,
     but no other edges among $\{u_1,u_2,v_1,v_2\}$ are present. Now we can add the chords $u_1s_1,u_2s_1$ to the cycle $s_1v_1u_1s_2u_2v_2s_1$. We delete $s_2$ and $d$. We have now lost at most four vertices from $B\cup D$: namely at most $v_1,v_2,d$ and one of the $u_i$ (if it was adjacent to $v_1$ and $v_2$ originally).

    \begin{figure}[!h]
        \centering
        \include{TikZ/fig:c100_1}
        %\section{$w2, x_2$ does not exist, $s_2$ has no nbr in $D$}
\scalebox{0.93}{ 
\begin{tikzpicture}[scale=0.5]
\node (before) at (0,0){
    \begin{tikzpicture}[scale=0.5]  
    \def\height{1.5};
    \def\width{2};
    
    \node[snode] (s_2) at (0, 0){$s_2$};
    \node[snode] (s_1) at (3*\width ,0){$s_1$};
    \node[snode] (s_3) at (6* \width, 0){$s_3$};
    
    \node[normal node] (v_1) at (2*\width, \height) {$v_1$};
    \node[normal node] (v_2) at (2*\width, -\height){$v_2$};
    \node[normal node] (w_1) at (4*\width, \height) {$w_1$};
    \node[normal node] (w_2) at (4*\width, -\height){$w_2$};
    
    \node[normal node] (u_1) at (1*\width, \height){$u_1$};
    \node[normal node] (u_2) at (1*\width, -\height){$u_2$};
    
    \node[normal node] (x_1) at (5*\width, 0){$x_1$};
   % \node[normal node] (x_2) at (5*\width, -1*\height){$x_2$};
    
    %edges
    \draw (s_1) to (v_1);
    \draw (s_1) to (v_2);
    \draw (v_1) to (u_1);
    \draw (v_2) to (u_2);
    \draw (s_2) to (u_1);
    \draw (s_2) to (u_2);
    \draw (s_3) to (x_1);
  %  \draw (s_3) to (x_2);
    \draw (w_1) to (x_1);
    \draw (w_2) to (x_1);
    \draw (s_1) to (w_1);
    \draw (s_1) to (w_2);

    \draw[maybe edge] (u_1) to (u_2);
    \draw[maybe edge] (u_1) to (v_2);
   \draw[maybe edge] (v_1) to (u_2);
   \draw[maybe edge] (v_1) to (v_2);

    \end{tikzpicture}
};
\node(after) at (before.east)[anchor=east,xshift=6cm]{
    \begin{tikzpicture}[scale=0.5]  
    \def\height{1.5};
    \def\width{2};
    
    \node[snode] (s_2) at (0, 0){$s_2$};
    %\node[snode] (s_1) at (3*\width ,0){$s_1$};
    \node[snode] (s_3) at (3* \width, 0){$s_3$};
    
   % \node[normal node] (v_1) at (2*\width, \height) {$v_1$};
 %   \node[normal node] (v_2) at (2*\width, -\height){$v_2$};
    %\node[normal node] (w_1) at (2*\width, \height) {$w_1$};
   % \node[normal node] (w_2) at (2*\width, -\height){$w_2$};
    
    \node[normal node] (u_1) at (1*\width, \height){$u_1$};
    \node[normal node] (u_2) at (1*\width, -\height){$u_2$};
    
    \node[normal node] (x_1) at (2*\width, 0){$x_1$};
   % \node[normal node] (x_2) at (5*\width, -1*\height){$x_2$};
    
    %edges
    \draw (u_1) to (s_2);
    \draw (u_2) to (s_2);

    %\draw[new edge, bend right = 15] (s_1) to (u_2);
    \draw (s_3) to (x_1);
  %  \draw (s_3) to (x_2);
    %\draw (w_1) to (x_1);
%    \draw (w_2) to (x_1);
    \draw[new edge] (u_1) to (u_2);
    \draw[new edge] (u_1) to (x_1);
    \draw[new edge] (u_2) to (x_1);
    
    \node[maybe new node] (z_1) at (-\width, \height) {$z_1$};
    \node[maybe new node] (z_2) at (-\width, -\height) {$z_2$};
    \draw[new edge] (z_1) to (u_1);
    \draw[new edge] (z_1) to (s_2);
    \draw[new edge] (z_2) to (s_2);
    \draw[new edge] (z_2) to (u_2);

    %\node[yellow node] (d) at (-\width, 0) {$d$};
    %\draw (s_2) to (d);

    \end{tikzpicture}
};
\draw [->, ultra thick] (before)--(after);
\end{tikzpicture}
}
        \caption{An illustration of the case $v_1, v_2, w_1, w_2, u_2$ all exist and are all distinct and $x_2$ does not exist. At top we show the case where $s_2$ has a neighbor $d \in D$. Some of the wavy edges may be present in $G$. As usual, there may be other edges present in $G$ that have not been drawn, but they are not relevant to our argument. At the bottom we illustrate the case where $s_2$ has no neighbor in $D$. For $i \in \{1,2 \}$, we add the vertices $z_i$ and edges $z_is_2, z_iu_i$ if needed to make $\text{deg}(u_i) \geq 4$ in $G'$. \textcolor{red}{Something is missing here.} }
        \label{fig:c100}
    \end{figure}
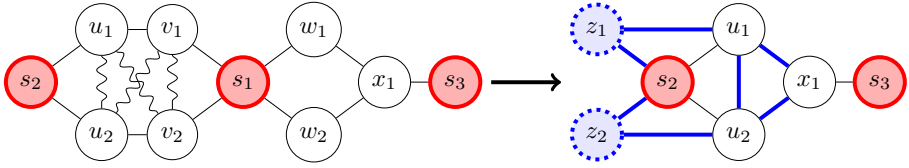
    
    \item Suppose now that $s_2$ has no neighbors in $D$. We remove $N[s_1]$ and add the edges $u_1x_1$,$u_2x_1$ and $u_1u_2$ if it is not already present. If $u_1$ has degree 3, then it has no neighbors outside of $u_2,x_1,s_2$ and so we may create a new vertex adjacent to both $u_1$ and $s_2$. Similarly, we can fix the degree of $u_2$ if needed. Note that since $s_1$ has degree at least four, it has no neighbors in $D$. We lose at most $v_1,v_2,w_1,w_2$ from $B\cup D$ and we found a smaller counterexample.
\end{itemize}
\subsection{The case when neither $u_2$ nor $x_2$ exists.} 
\label{app:neither-u2-nor-x2-exists}
If both $u_1$ and $x_1$ do not have degree exactly four, then we can remove $N[s_1]$ and add the edge $u_1x_1$; in this case we only lose a subset of $\{v_1,v_2,w_1,w_2\}$ from $B\cup D$. Hence we can assume by symmetry that $u_1$ has degree exactly four.
    \begin{itemize}
    \item We first handle the case in which $s_2$ has no neighbor in $D$. Since $u_1$ has degree exactly $4$, after removing $N[s_1]$ we can create a new vertex $v$ and add the edges $u_1v,s_2v,x_1v,u_1x_1$. As a result, we have lost at most $v_1,v_2,w_1,w_2$ from $B\cup D$ and found a smaller counterexample. See Figure \ref{fig:No_u2w2_s2_has_no_Dnbrs}. 
    \item Suppose now that $s_2$ has only neighbors in $D \cup \{u_1 \}$, which we name $d_1,d_2$ (where $d_2$ may or may not exist). We remove $N[s_2]\setminus \{u_1\}$ (at most three vertices), remove the edge $v_1v_2$ and add the edge $u_1s_1$. We again found a smaller counterexample as the only vertices we may have lost from $B\cup D$ are $v_1,v_2,d_1,d_2$.
    \item Suppose now that $s_2$ has exactly one neighbor $d\in D$. It may have another neighbor $y \neq u_1,d$, which if it exists, is not in $D$. We delete the vertices $s_2,d$ as well as the edges $u_1v_1$ and $v_1v_2$ (if these exist). As $u_1$ was a cut-vertex previously, we can now add the edges $u_1s_1$ and $ys_1$ (say along the path $u_1v_2s_1$) to ensure that the size of the dominating set has dropped by one whereas we lost at most $d,u_1,v_1,v_2$ from $B\cup D$. See Figure \ref{fig:No_w2x2_s2HasExactlyOneDNbr}.
\end{itemize}
\begin{figure}
    \centering
    \scalebox{0.93}{ 
\begin{tikzpicture}[scale=0.5]
\node (before) at (0,0){
    \begin{tikzpicture}[scale=0.5]  
    \def\height{1.5};
    \def\width{2};
    
    \node[snode] (s_2) at (0, 0){$s_2$};
    \node[snode] (s_1) at (3*\width ,0){$s_1$};
    \node[snode] (s_3) at (6* \width, 0){$s_3$};
    
    \node[normal node] (v_1) at (2*\width, \height) {$v_1$};
    \node[normal node] (v_2) at (2*\width, -\height){$v_2$};
    \node[normal node] (w_1) at (4*\width, \height) {$w_1$};
    \node[normal node] (w_2) at (4*\width, -\height){$w_2$};
    
    \node[normal node] (u_1) at (1*\width, 0){$u_1$};

    \node[normal node] (x_1) at (5*\width, 0){$x_1$};
   % \node[normal node] (x_2) at (5*\width, -1*\height){$x_2$};
    
    %edges
    \draw (s_1) to (v_1);
    \draw (s_1) to (v_2);
    \draw (v_1) to (u_1);
    \draw (v_2) to (u_1);
    \draw (s_2) to (u_1);
    \draw (s_2) to (u_1);
    \draw (s_3) to (x_1);
  %  \draw (s_3) to (x_2);
    \draw (w_1) to (x_1);
    \draw (w_2) to (x_1);
    \draw (s_1) to (w_1);
    \draw (s_1) to (w_2);

    \end{tikzpicture}
};
\node(after) at (before.east)[anchor=east,xshift=6cm]{
    \begin{tikzpicture}[scale=0.5]  
        \def\height{2.5};
        \def\width{2.5};
        \node[snode] (s_2) at (0,0){$s_2$};
        \node[normal node] (u_1) at (\width, 0){$u_1$};
        \node[normal node] (x_1) at (2*\width, 0){$x_1$};
        \node[snode] (s_3) at (3*\width, 0){$s_3$};
        \draw (s_2) to (u_1);
        \draw (s_3) to (x_1);
        \draw[new edge] (u_1) to (x_1);
        
        \node[special node] (v) at (\width, \height) {$v$};
        \draw[new edge] (s_2) to (v);
        \draw[new edge] (u_1) to (v);
        \draw[new edge] (x_1) to (v);

    %\node[yellow node] (d) at (-\width, 0) {$d$};
    %\draw (s_2) to (d);

    \end{tikzpicture}
};
\draw [->, ultra thick] (before)--(after);
\end{tikzpicture}
}
    \caption{An illustration of the reduction for the case where $u_2, w_2$ do not exist and $s_2$ has no neighbors in $D$. We only draw edges that are relevant to our argument.}
    \label{fig:No_u2w2_s2_has_no_Dnbrs}
\end{figure}
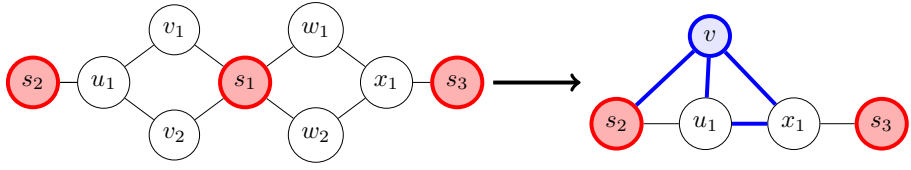

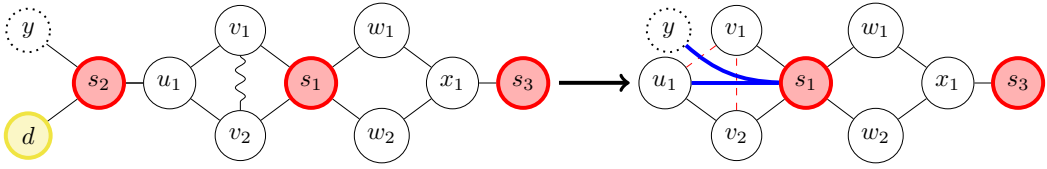
\begin{figure}
    \centering
    %\section{no $u_2, w_2$, $s_2$ has a single neighbour in $D$}
%stuff
\scalebox{0.93}{ 
\begin{tikzpicture}[scale=0.5]
\node (before) at (0,0){
    \begin{tikzpicture}[scale=0.5]  
    \def\height{1.5};
    \def\width{2};
    
    \node[snode] (s_2) at (0, 0){$s_2$};
    \node[snode] (s_1) at (3*\width ,0){$s_1$};
    \node[snode] (s_3) at (6* \width, 0){$s_3$};
    
    \node[normal node] (v_1) at (2*\width, \height) {$v_1$};
    \node[normal node] (v_2) at (2*\width, -\height){$v_2$};
    \node[normal node] (w_1) at (4*\width, \height) {$w_1$};
    \node[normal node] (w_2) at (4*\width, -\height){$w_2$};
    
    \node[normal node] (u_1) at (1*\width, 0){$u_1$};

    \node[normal node] (x_1) at (5*\width, 0){$x_1$};
   % \node[normal node] (x_2) at (5*\width, -1*\height){$x_2$};
    
    %edges
    \draw (s_1) to (v_1);
    \draw (s_1) to (v_2);
    \draw (v_1) to (u_1);
    \draw (v_2) to (u_1);
    \draw (s_2) to (u_1);
    \draw (s_2) to (u_1);
    \draw (s_3) to (x_1);
  %  \draw (s_3) to (x_2);
    \draw (w_1) to (x_1);
    \draw (w_2) to (x_1);
    \draw (s_1) to (w_1);
    \draw (s_1) to (w_2);
    
    \node[yellow node] (d) at (-\width, -\height){$d$};
    \draw(s_2) to (d);
    \draw[maybe edge] (v_1) to (v_2);
    \node[maybe node] (y) at (-\width, \height){$y$};
    \draw (s_2) to (y);
    \end{tikzpicture}

};
\node(after) at (before.east)[anchor=east,xshift=7cm]{
    \begin{tikzpicture}[scale=0.5]  
    \def\height{1.5};
    \def\width{2};
    
    %\node[snode] (s_2) at (0, 0){$s_2$};
    \node[snode] (s_1) at (3*\width ,0){$s_1$};
    \node[snode] (s_3) at (6* \width, 0){$s_3$};
    
    \node[normal node] (v_1) at (2*\width, \height) {$v_1$};
    \node[normal node] (v_2) at (2*\width, -\height){$v_2$};
    \node[normal node] (w_1) at (4*\width, \height) {$w_1$};
    \node[normal node] (w_2) at (4*\width, -\height){$w_2$};
    
    \node[normal node] (u_1) at (1*\width, 0){$u_1$};

    \node[normal node] (x_1) at (5*\width, 0){$x_1$};
   % \node[normal node] (x_2) at (5*\width, -1*\height){$x_2$};
    
    %edges
    \draw (s_1) to (v_1);
    \draw (s_1) to (v_2);
    \draw[non edge] (v_1) to (u_1);
    \draw[non edge] (v_2) to (v_1);
    \draw (v_2) to (u_1);
    \draw (s_3) to (x_1);
  %  \draw (s_3) to (x_2);
    \draw (w_1) to (x_1);
    \draw (w_2) to (x_1);
    \draw (s_1) to (w_1);
    \draw (s_1) to (w_2);

    \node[maybe node] (y) at (\width, \height){$y$};
    \draw[new edge, bend left = 20] (s_1) to (y);
    \draw[new edge] (u_1) to (s_1);

    \end{tikzpicture}
};
\draw [->, ultra thick] (before)--(after);
\end{tikzpicture}
}
    \caption{An illustration of the reduction for the case where $u_2, w_2$ do not exist and $s_2$ has exactly one neighbor $d \in D$. We use dashed red lines to emphasize that $v_1$ is not adjacent to $u_1, u_2$ in $G'$ and a wavy line to show that $v_1$ may be adjacent to $v_2$ in $G$. As usual, we do not draw edges that are not relevant to our argument.}
    \label{fig:No_w2x2_s2HasExactlyOneDNbr}
\end{figure}
\end{appendices}

\bibliographystyle{plainurl}

\end{document}